\documentclass[12pt]{iopart}
\usepackage{iopams} 
\usepackage{mathbbol} 
\usepackage[colorlinks, linkcolor =blue]{hyperref} 
\usepackage{amsthm} 
\usepackage{mathrsfs} 
\usepackage{graphicx} 
\usepackage{array} 

\usepackage{epsf}
\usepackage{subfigure}
\usepackage{epstopdf}
\newcommand{\bm}[1]{\boldsymbol{#1}} 
\newcommand{\binom}[2]{{#1 \choose #2}} 
\newtheorem{lemma}{Lemma} 
\newtheorem*{lemma*}{Lemma} 
\newtheorem{theorem}{Theorem} 
\newtheorem*{proposition*}{Proposition} 
\newtheorem*{corollary*}{Corollary} 

\newcolumntype{M}[1]{>{\centering\arraybackslash}m{#1}}
\newcolumntype{N}{@{}m{0pt}@{}}

\begin{document}


\title{Microreversibility, nonequilibrium current fluctuations, and response theory}

\author{M Barbier and P Gaspard}

\address{Center for Nonlinear Phenomena and Complex Systems, Universit\'e libre de Bruxelles (ULB), Code Postal 231, Campus Plaine, B-1050 Brussels, Belgium}

\ead{Maximilien.Barbier@ulb.ac.be,gaspard@ulb.ac.be}


\begin{abstract}
Microreversibility rules the fluctuations of the currents flowing across open systems in nonequilibrium (or equilibrium) steady states. As a consequence, the statistical cumulants of the currents and their response coefficients at arbitrary orders in the deviations from equilibrium obey time-reversal symmetry relations.  It is shown that these relations allow us to systematically reduce the amount of independent quantities that need to be measured experimentally or computed theoretically in order to fully characterize the linear and nonlinear transport properties of general open systems.  This reduction is shown to approach one half for quantities of arbitrarily high orders.
\end{abstract}

\vspace{2pc}

\section{Introduction}\label{intro_sec}

The motion of electrons and other particles composing matter is ruled by microscopic Hamiltonian dynamics that is symmetric under the time-reversal transformation.  This symmetry called microreversibility has fundamental consequences for the transport properties of matter from micro- to macro-scales.  At positive temperatures, the motion of particles is ceaseless because of thermal fluctuations.  Yet, there are no net currents of energy or particles flowing through matter at thermodynamic equilibrium where the principle of detailed balance holds \cite{deGr,Diu,Kle55}.  To generate such currents, the system of interest should be driven out of equilibrium by external forces or contacts with external reservoirs at different temperatures or chemical potentials, which has the effect of breaking detailed balance \cite{Kle55}.  Under such nonequilibrium conditions, the generated currents of energy or particles depend on the mechanical or thermodynamic forces, commonly called affinities \cite{deGr,DeDon,Prig,Cal}, which are driving the system away from equilibrium. Close to equilibrium, the dependence of the currents on the affinities is linear.  In this regime, the linear response coefficients satisfy the Green-Kubo formulae \cite{Gre52,Gre54,Kub57}, the fluctuation-dissipation theorem \cite{CW51, Kub66}, as well as the Onsager-Casimir reciprocity relations \cite{Ons31a,Ons31b,Cas45} as a consequence of microreversibility. 

However, many systems are driven farther away from equilibrium in regimes where the currents have nonlinear dependences on the affinities.  Such dependences arise due to heterogeneities met by the particles in their motion across the system, as is the case for instance in semiconducting electronic devices.  A fundamental issue is to understand the consequences of microreversibility in these nonlinear regimes.  To this respect, remarkable results have been obtained in the form of the so-called \textit{fluctuation theorems} or \textit{fluctuation relations} (FR) \cite{ECM93,ES94,GC95,Jar97,Kur98,LS99,Cro99,ZC03,AGM09,EHM09,CHT11,HPPG11,Gas13_1, Gas13_2}.  Their main interest is that they are valid in nonlinear as well as linear regimes.  In these latter regimes close to equilibrium, FR can be used to recover the Onsager-Casimir reciprocity relations and the Green-Kubo formulae \cite{Gal96,LS99}. Beyond, it has been shown that
general Onsager-Casimir-like identities between the nonlinear response coefficients and the statistical cumulants of the fluctuating currents can be directly deduced from the FR \cite{AGM09,AG04,AG07,SU08,WF15}. These identities allow us to reduce the number of coefficients we need to measure experimentally or compute theoretically in order to characterize transport in nonlinear regimes.

In this paper, our purpose is to determine quantitatively the amount of reduction provided by microreversibility in the study of linear and nonlinear transport properties.  We consider general nonequilibrium systems crossed by heat or particle currents \cite{JW04,AG06}, or undergoing chemical reactions, the rates of which are the currents from reactants to products \cite{AG04,Gas04}.  These systems can be described by classical or quantum mechanics, or the theory of stochastic processes \cite{AGM09,EHM09,Gas13_1}. They are supposed to be maintained in nonequilibrium (or equilibrium) steady states by contacts with arbitrarily large reservoirs of energy and particles at fixed values of temperatures and chemical potentials, in the absence of magnetic field. Our study is then entirely built on the FR expressed in terms of the generating function of the statistical cumulants of the currents. These latter provide the full characterization of the current statistics, which is also known as full counting statistics, in particular for currents of electrons \cite{Naz}.

We first deduce from the FR a set of relations between the cumulants and their responses to the nonequilibrium constraints. The main outcome of our work is then to show that the elements of this set are not all independent and to identify the relations that can be deduced from the others. This result is used to determine, as a consequence of microreversibility, the number of independent quantities that describe the transport properties. In particular, we show that the FR effectively divides by two the total number of quantities that need to be known in order to specify the full current statistics in a nonequilibrium steady state.

The paper is structured as follows. Section~\ref{fluct_rel_sec} stands as a brief introduction to the FR. The statistical cumulants and their responses to the nonequilibrium constraints are then defined in section~\ref{coef_sec}, where we also see how the FR generates relations between the latter quantities. Section~\ref{indep_rel_sec} contains the core of our study. There we show that some of the aforementioned relations between the cumulants and their responses are trivial, in the sense that they can be deduced from a subset of these relations. It is also worth pointing out that we deduce, as a direct by-product of our mathematical analysis, a general identity satisfied by coefficients of Euler polynomials. We then illustrate our results in section~\ref{number_ind_coef_sec}, where we quantify the total number of independent quantities that need to be known to fully characterize the nonlinear response of the system to its nonequilibrium constraints. Concluding remarks are finally drawn in section~\ref{conclusion_sec}.


\section{Fluctuation relation}\label{fluct_rel_sec}

Following the discussion of \cite{AGM09}, we consider an open system coupled to $r\geqslant 2$ arbitrarily large reservoirs of energy and particles at the temperatures $T_j$ and chemical potentials $\mu_j$ ($j = 1,2,\ldots, r$).  We assume that the coupling does not change the thermodynamic properties of the reservoirs, so that $T_j$ and $\mu_j$ are all time independent.  Upon differences in the temperatures and chemical potentials, the system is driven out of equilibrium and subjected to currents of energy and particles.  After some relaxation time, the system reaches a nonequilibrium steady state characterized by the thermal and chemical affinities
\begin{eqnarray}
A_{j \mathrm{E}} \equiv \frac{1}{k_{\mathrm{B}} T_r} - \frac{1}{k_{\mathrm{B}} T_j} \qquad\mbox{and}\qquad A_{j \mathrm{N}} \equiv \frac{\mu_j}{k_{\mathrm{B}} \, T_j} - \frac{\mu_r}{k_{\mathrm{B}} \, T_r} \, ,
\label{therm_chem_aff_def}
\end{eqnarray}
which are defined with respect to some reference reservoir, here taken as $j=r$.  In~\eref{therm_chem_aff_def}, $k_{\mathrm{B}}$~denotes Boltzmann's constant.  At equilibrium, these affinities are vanishing together with the mean currents. In the sequel, we collectively denote the affinities~\eref{therm_chem_aff_def} by the vector $\bm{A}$.

The microscopic dynamics of the process is ruled during the time interval $[0,t]$ by the total Hamiltonian $\hat H$, which is assumed to satisfy the symmetry
\begin{eqnarray}
\hat\Theta \hat H \hat\Theta^{-1} = \hat H
\label{microrev_Ham}
\end{eqnarray}
under the anti-unitary time-reversal operator $\hat\Theta$. The symmetry~\eref{microrev_Ham} is the expression of microreversibility. Here, we suppose that there is no external magnetic field.

In order to characterize the fluctuations of the currents, a full current statistics is carried out for the exchanges of energy and particles between the reservoirs. The amounts of energy and particles exchanged between the $j^{\mathrm{th}}$ reservoir and the reference one are denoted by $\Delta\bm{x}=\{\Delta x_k\}_{k=1}^{\chi}$, where $\Delta x_k$ is either an energy exchange $\Delta\epsilon_j$ or a particle exchange $\Delta \nu_j$ and $\chi$ is the total number of currents.  The statistics of these exchanges during the time interval $[0,t]$ can be described by the probability density $p_t\left( \Delta\bm{x} \right)$ or, equivalently, by the generating function (GF) of its statistical moments \cite{vanK} defined by the Laplace-type integral
\begin{eqnarray}
G_t \left( \bm{\lambda} \right) \equiv \int p_t ( \Delta\bm{x}) \exp( -\bm{\lambda}\cdot\Delta\bm{x}) \, d\Delta\bm{x}
\label{mom_GF_def}
\end{eqnarray}
 where $\bm{\lambda}=\{\lambda_k\}_{k=1}^{\chi}$ are commonly referred to as counting parameters (or counting fields).  There are as many counting parameters as nonvanishing affinities.  For instance, we would have $\chi=r-1$ counting parameters for one species of particles, e.g. electrons, flowing between $r$ reservoirs at different chemical potentials in an isothermal system.

The two-point quantum measurement scheme \cite{AGM09,EHM09,Kur00} allows us to express the GF \eref{mom_GF_def} according to
\begin{eqnarray}
G_t \left( \bm{\lambda} \right) = {\rm tr} \, \hat\rho_0\, {\rm e}^{i\hat H t} \, {\rm e}^{-\bm{\lambda}\cdot\hat{\bm X}} \, {\rm e}^{-i\hat H t} \, {\rm e}^{\bm{\lambda}\cdot\hat{\bm X}} 
\label{mom_GF_quant}
\end{eqnarray}
where $\hat{\bm X}=\{\hat X_k\}$ collectively denotes the Hamiltonian operators $\hat H_j$ and particle number operators $\hat N_j$ of the reservoirs themselves. The initial density operator $\hat\rho_0$ describes the grand-canonical ensemble for the reservoirs at the given temperatures and chemical potentials:
\begin{eqnarray}
\hat\rho_0 = \prod_{j=1}^{r} \frac{1}{\Xi_j} \, {\rm e}^{-\left(\hat H_j-\mu_j\hat N_j\right)/(k_{\rm B}T_j)}
\label{rho0}
\end{eqnarray}
where $\Xi_j$ denotes the corresponding partition function.

In the long-time limit $t \to \infty$, we introduce the cumulant GF according to
\begin{eqnarray}
Q \left( \bm{\lambda} , \bm{A} \right) \equiv - \lim_{t \to \infty} \frac{1}{t} \, \ln G_t \left( \bm{\lambda} \right)
\label{cumulant_GF_def}
\end{eqnarray}
where the dependence on the affinities $\bm{A}$ is explicitly denoted and arises from the parameters in the density operator~(\ref{rho0}). The cumulant GF characterizes the full current statistics in the equilibrium or nonequilibrium steady state that is reached in the long-time limit $t \to \infty$.  For ballistic transport in mesoscopic circuits, this theory is related to the scattering theory of transport \cite{Gas13_1,Naz,L57,L70,B86a,B86b,LL93,LS11}.

It can then be shown \cite{AGM09,Gas13_1,Gas13_2} that, as a consequence of the time-reversal symmetry~\eref{microrev_Ham} the cumulant GF~\eref{cumulant_GF_def} satisfies the symmetry
\begin{eqnarray}
Q \left( \bm{\lambda} , \bm{A} \right) = Q \left( \bm{A} - \bm{\lambda} , \bm{A} \right) .
\label{fluct_rel}
\end{eqnarray}
This relation is known as a (multivariate) \textit{fluctuation relation} (FR). As is well known, a FR can be shown \cite{LS99,AGM09,Gal96,AG04,AG07,SU08,WF15} to generate the Green-Kubo formulae \cite{Gre52,Gre54,Kub57} and the Onsager-Casimir reciprocity relations \cite{Ons31a,Ons31b,Cas45}. Such relations are valid in the linear response regime, i.e. for a nonequilibrium system close to equilibrium. The fundamental interest of the FR~\eref{fluct_rel} is that it remains valid arbitrarily far from equilibrium. Therefore, a FR also gives access to the properties of the nonlinear response regime, as we now discuss.


\section{Statistical cumulants and response coefficients}\label{coef_sec}

Here we define the statistical cumulants and response coefficients, that are related to the GF~\eref{cumulant_GF_def} through derivatives with respect to its independent variables $\bm{\lambda}$ and $\bm{A}$.  By construction, the statistical cumulants are defined from the Taylor expansion of~\eref{cumulant_GF_def} in powers of the counting parameters $\bm{\lambda}$. We denote by $Q_{\alpha_1 \cdots \alpha_m} \left( \bm{A} \right)$ the $m^{\mathrm{th}}$~cumulant, and we have
\begin{eqnarray}
Q_{\alpha_1 \cdots \alpha_m} \left( \bm{A} \right) \equiv \left. \frac{\partial^m Q}{\partial \lambda_{\alpha_1} \cdots \partial \lambda_{\alpha_m}} \left( \bm{\lambda}, \bm{A} \right) \right|_{\bm{\lambda} = \bm{0}} \equiv \frac{\partial^m Q}{\partial \lambda_{\alpha_1} \cdots \partial \lambda_{\alpha_m}} \left( \bm{0}, \bm{A} \right)
\label{m_cumulant_def}
\end{eqnarray}
where $\alpha_j$, $j = 1, \ldots , m$, can take any value between 1 and $\chi$. Note that the cumulants can also be defined for $m=0$ because of the normalization condition
\begin{eqnarray}
Q \left( \bm{0} , \bm{A} \right) = 0 \, ,
\label{cumulant_GF_norm_cond}
\end{eqnarray}
which can be readily obtained from~\eref{mom_GF_def} and~\eref{cumulant_GF_def}. It must thus be understood that the $0^{\mathrm{th}}$ cumulant merely vanishes. Therefore, we can write the cumulant GF as
\begin{eqnarray}
Q \left( \bm{\lambda} , \bm{A} \right) = \sum_{m = 0}^{\infty} \frac{1}{m!} \, Q_{\alpha_1 \cdots \alpha_m} \left( \bm{A} \right) \lambda_{\alpha_1} \cdots \lambda_{\alpha_m}
\label{Q_exp_count_par}
\end{eqnarray}
where we use, here and in the sequel, Einstein's summation convention for repeated indices.

Now, the cumulants~\eref{m_cumulant_def} are by construction functions of the affinities. We can thus expand them as power series of $\bm{A}$, i.e.
\begin{eqnarray}
Q_{\alpha_1 \cdots \alpha_m} \left( \bm{A} \right) = \sum_{n = 0}^{\infty} \frac{1}{n!} \, Q_{\alpha_1 \cdots \alpha_m \, , \, \beta_1 \cdots \beta_n} A_{\beta_1} \cdots A_{\beta_n}
\label{cumulant_exp_aff}
\end{eqnarray}
for any $m \geqslant 0$, with
\begin{eqnarray}
\fl Q_{\alpha_1 \cdots \alpha_m \, , \, \beta_1 \cdots \beta_n} \equiv \frac{\partial^n Q_{\alpha_1 \cdots \alpha_m}}{\partial A_{\beta_1} \cdots \partial A_{\beta_n}} \left( \bm{0} \right) = \frac{\partial^{m+n} Q}{\partial \lambda_{\alpha_1} \cdots \partial \lambda_{\alpha_m} \partial A_{\beta_1} \cdots \partial A_{\beta_n}} \left( \bm{0}, \bm{0} \right) .
\label{m_cumulant_n_resp_def}
\end{eqnarray}
Here it must be understood that (i) for $m=0$, we have $Q_{\, , \, \beta_1 \cdots \beta_n} = 0$ in view of the normalization condition~\eref{cumulant_GF_norm_cond}, while (ii) for $n=0$, we merely have $Q_{\alpha_1 \cdots \alpha_m \, , \,} = Q_{\alpha_1 \cdots \alpha_m} \left( \bm{0} \right)$. Therefore, substituting the expression~\eref{cumulant_exp_aff} of the cumulant $Q_{\alpha_1 \cdots \alpha_m}$ into~\eref{Q_exp_count_par}, we obtain the following expansion of the cumulant GF as a power series of both the counting parameters $\bm{\lambda}$ and the affinities $\bm{A}$:
\begin{eqnarray}
Q \left( \bm{\lambda} , \bm{A} \right) = \sum_{m , n = 0}^{\infty} \frac{1}{m! n!} \, Q_{\alpha_1 \cdots \alpha_m \, , \, \beta_1 \cdots \beta_n} \lambda_{\alpha_1} \cdots \lambda_{\alpha_m} A_{\beta_1} \cdots A_{\beta_n}
\label{Q_exp_count_par_and_aff}
\end{eqnarray}
where we emphasize that, using Einstein's convention, the indices $\alpha_j$ and $\beta_j$ are all summed from 1 to $\chi$ ($\chi$ being the total number of independent currents).

It should be noted that the quantity $Q_{\alpha_1 \cdots \alpha_m \, , \, \beta_1 \cdots \beta_n}$ is invariant under any permutation of the subscripts that appear on one definite side of the comma, i.e.
\begin{eqnarray}
Q_{\alpha_1 \cdots \alpha_m \, , \, \beta_1 \cdots \beta_n} = Q_{\alpha_{P_m (1)} \cdots \alpha_{P_m (m)} \, , \, \beta_{P_n (1)} \cdots \beta_{P_n (n)}}
\label{Q_inv_perm_indices}
\end{eqnarray}
for any permutations $P_m$ and $P_n$ of $m$ and $n$ elements, respectively. This invariance is rooted in the definition~\eref{m_cumulant_n_resp_def} of $Q_{\alpha_1 \cdots \alpha_m \, , \, \beta_1 \cdots \beta_n}$. Indeed, note that each of the $m$ subscripts $\alpha_i$ on the left of the comma denotes a (partial) derivative of the cumulant GF with respect to a counting parameter $\lambda_i$. Such derivatives can be performed in an arbitrary order. The same property holds for the $n$ subscripts $\beta_i$ on the right of the comma, as each of them denotes a (partial) derivative of the GF with respect to an affinity $A_i$.

As is clear on~\eref{therm_chem_aff_def}, to have $\bm{A} = \bm{0}$ describes the equilibrium state of the $r$ reservoirs, in which case there exist no net currents. The quantity $Q_{\alpha_1 \cdots \alpha_m \, , \, \beta_1 \cdots \beta_n}$ characterizes the response of the $m^{\mathrm{th}}$ cumulant $Q_{\alpha_1 \cdots \alpha_m}$ with respect to the affinities $A_{\beta_1} , \ldots , A_{\beta_n}$  about the equilibrium value. In particular, the response $Q_{\alpha \, , \, \beta_1 \cdots \beta_n}$ of the mean current $Q_{\alpha}$ defines a so-called $n^{\rm th}$-order \textit{response coefficient}. The response coefficients characterize the response of the mean current  $Q_{\alpha}$ with respect to some nonequilibrium constraint imposed by the affinities $\bm{A}$. The first-order response coefficients $Q_{\alpha \, , \, \beta}$ characterize the \textit{linear response} properties of the system, and are precisely the ones that satisfy the Green-Kubo formulae \cite{Gre52,Gre54,Kub57} and the Onsager-Casimir reciprocity relations \cite{Ons31a,Ons31b,Cas45}. The higher-order response coefficients $Q_{\alpha \, , \, \beta_1 \cdots \beta_n}$, with $n>1$, characterize the \textit{nonlinear response} properties of the system. As is well known \cite{AGM09,HPPG11,AG04,AG07,SU08,WF15}, such nonlinear response coefficients turn out to satisfy Onsager-Casimir-like relations as a direct consequence of the FR~\eref{fluct_rel}.

A first strategy to derive general relations satisfied by the response coefficients is to successively differentiate the FR~\eref{fluct_rel} with respect to both the counting parameters $\bm{\lambda}$ and the affinities $\bm{A}$ \cite{AGM09,AG04}. A more global approach is to infer from~\eref{fluct_rel} a compact relation satisfied by the coefficients $Q_{\alpha_1 \cdots \alpha_m \, , \, \beta_1 \cdots \beta_n}$ \cite{AG07,SU08,WF15}. This can be for instance done by (i) expanding both sides of the FR~\eref{fluct_rel} as power series of $\bm{\lambda}$ and $\bm{A}$, and (ii) identifying the coefficients of a same power of both the counting parameters and the affinities.

To this end, it proves useful to write the FR~(\ref{fluct_rel}) by making the substitution $\bm{\lambda}\to-\bm{\lambda}$ as
\begin{eqnarray}
Q \left(-\bm{\lambda} , \bm{A} \right) = Q \left(\bm{\lambda}+ \bm{A} , \bm{A} \right) =\hat T(\bm{A}) Q \left(\bm{\lambda} , \bm{A} \right) 
\label{fluct_rel_2}
\end{eqnarray}
in terms of the translation operator
\begin{eqnarray}
\hat T \left( \bm{A} \right) \equiv {\rm e}^{{\bm{A} \cdot \frac{\partial}{\partial \bm{\lambda}}}} \, .
\label{transl}
\end{eqnarray}
Therefore, expanding both sides of~\eref{fluct_rel_2} as power series of $\bm{\lambda}$ and $\bm{A}$, we can show that the FR yields the general relation
\begin{eqnarray}
Q_{\alpha_1 \cdots \alpha_m \, , \, \beta_1 \cdots \beta_n} = (-1)^m \sum_{j=0}^{n} Q_{\alpha_1 \cdots \alpha_m \, , \, \beta_1 \cdots \beta_n}^{(j)}  \qquad\qquad \forall \, m,n \geqslant 0 
\label{gen_rel_response_coef}
\end{eqnarray}
between the quantities $Q_{\alpha_1 \cdots \alpha_m \, , \, \beta_1 \cdots \beta_n}$, with $Q^{(0)} \equiv Q$ and
\begin{eqnarray}
Q_{\alpha_1 \cdots \alpha_m \, , \, \beta_1 \cdots \beta_n}^{(j)} \equiv \sum_{k_1 = 1}^{n} \sum_{k_{2}=1 \atop k_{2} > k_{1}}^{n} \cdots \sum_{k_{j}=1 \atop k_{j} > k_{j-1}}^{n} Q_{\alpha_1 \cdots \alpha_m \beta_{k_1} \cdots \beta_{k_j} \, , \, (\bm{\cdot})}
\label{Q_j_expr}
\end{eqnarray}
for $j \geqslant 1$, where $(\bm{\cdot})$ denotes the set of all subscripts $\beta_k$ that are different of the subscripts $\beta_i$ present on the left of the comma (i.e. $\beta_{k_1} , \ldots , \beta_{k_j}$ here). The details of this calculation may be found in~\ref{proof_rel_app} where we present a derivation alternative to \cite{AG07,AndPhD}.

The interest of the result~\eref{gen_rel_response_coef} is that it gives a direct access to relations satisfied by the response coefficients of arbitrary order. This can be seen by investigating all possible relations obtained from~\eref{gen_rel_response_coef} for a fixed total number $\mathcal{N} = m+n$ of subscripts. To consider such a fixed $\mathcal{N}$ is justified by the structure of~\eref{gen_rel_response_coef}, for the latter only involves quantities $Q_{\alpha_1 \cdots \alpha_m \, , \, \beta_1 \cdots \beta_n}$ that all possess the same total number of subscripts. We note that the Green-Kubo and Onsager reciprocity relations are given by~\eref{gen_rel_response_coef} for $\mathcal{N} = m+n=1$, while \cite{AG04,AG07,AG06} explicitly give the relations up to $\mathcal{N} = m+n\leqslant 3$.  Here, we propose to determine the relations that are independent among those obtained from~\eref{gen_rel_response_coef} for an arbitrary fixed number $\mathcal{N} = m+n$ and all possible values of~$m$ and~$n$. By proving in section~\ref{indep_rel_sec} the statement of \cite{AG07} that any relation obtained from~\eref{gen_rel_response_coef} for an even index $m$ can be deduced from the relations~\eref{gen_rel_response_coef} for odd indices $m$, we complete the analysis initiated in \cite{AG07}. We hence show that the independent relations are the ones obtained from~\eref{gen_rel_response_coef} for any odd value of the index $m$. We then use this result in section~\ref{number_ind_coef_sec} to count the number of independent quantities $Q_{\alpha_1 \cdots \alpha_m \, , \, \beta_1 \cdots \beta_n}$.


\section{Independent relations}\label{indep_rel_sec}

An essential point is that the relations~\eref{gen_rel_response_coef} separate into two series of relations depending on the parity of the integer $m$,
\begin{eqnarray}
m~\mbox{even}:\qquad  &&0 = \sum_{j=1}^{n} Q_{\alpha_1 \cdots \alpha_m \, , \, \beta_1 \cdots \beta_n}^{(j)} \label{rel_m_even}\\
m~\mbox{odd}:\qquad  &&Q_{\alpha_1 \cdots \alpha_m \, , \, \beta_1 \cdots \beta_n} = -\frac{1}{2} \sum_{j=1}^{n} Q_{\alpha_1 \cdots \alpha_m \, , \, \beta_1 \cdots \beta_n}^{(j)} \label{rel_m_odd}
\end{eqnarray}
and that these relations are not all independent of each other.  In this section, our aim is to prove the following:

\begin{theorem}
The relations~\eref{rel_m_even} for any even index $m$ can be deduced from the relations~\eref{rel_m_odd} corresponding to odd indices $m$.
\label{theorem1}
\end{theorem}

\noindent We then obtain as an immediate consequence of this theorem that

\begin{corollary*}
Among all the relations~\eref{gen_rel_response_coef}, the set of independent ones is given by~\eref{rel_m_odd} corresponding to odd indices $m$.
\label{corollary}
\end{corollary*}

We consider all the relations that are obtained from~\eref{rel_m_even} and~\eref{rel_m_odd} for a given arbitrary total number $\mathcal{N} = m+n \geqslant 1$ of subscripts. We begin with an arbitrary, but fixed even index $m$ taking
\begin{eqnarray}
m=m_K \equiv 2 K \qquad\mbox{and}\qquad n=n_K = \mathcal{N} - m_K = \mathcal{N} - 2 K 
\label{m+n_K_def}
\end{eqnarray}
where $K$ is an arbitrary integer such that $0 \leqslant K \leqslant \mathbb{E} \left( \mathcal{N} / 2 \right)$. Here and in the sequel $\mathbb{E} (x)$ denotes the integer part of the positive real number $x$, i.e. the natural number $k \geqslant 0$ such that $k \leqslant x < k+1$. 

With the notations
\begin{eqnarray}
\bm{\alpha} \equiv \alpha_1 \cdots \alpha_{m_K} \qquad\mbox{and}\qquad \bm{\beta} \equiv \beta_1 \cdots \beta_{n_K}
\label{alpha_def}
\end{eqnarray}
the relation~\eref{rel_m_even} for $m=m_K$ and $n=n_K$ reads
\begin{eqnarray}
0 = \sum_{j=1}^{n_K} Q_{\bm{\alpha} \, , \, \bm{\beta}}^{(j)}
\label{rel_m_K}
\end{eqnarray}
where the quantities $Q^{(j)}$ are defined by~\eref{Q_j_expr}, i.e. here
\begin{eqnarray}
Q_{\bm{\alpha} \, , \, \bm{\beta}}^{(j)} = \sum_{k_1} \sum_{k_{2} > k_{1}} \cdots \sum_{k_{j} > k_{j-1}} Q_{\bm{\alpha} \beta_{k_1} \cdots \beta_{k_j} \, , \, (\boldsymbol{\cdot})}
\label{Q_j_expr_remind}
\end{eqnarray}
for any $1 \leqslant j \leqslant n_K$. We have also introduced the notation
\begin{eqnarray}
\sum_{k} \equiv \sum_{k=1}^{n_K} \, ,
\label{notation_sums}
\end{eqnarray}
which will prove useful in the sequel to lighten some of our expressions.

Then we consider all the possible odd values of the index $m$ such that $m > m_K = 2 K$. That is, we consider all the possible indices 
\begin{eqnarray}
m = m'_k\equiv 2 k + 1 \qquad\mbox{and}\qquad n = n'_k =\mathcal{N} - m'_k = \mathcal{N} - 2 k - 1
\label{prime_indices}
\end{eqnarray}
where $k$ is an arbitrary integer such that $K \leqslant k \leqslant \mathbb{E} \left[ \left( \mathcal{N} + 1 \right) / 2 \right] - 1$. Since $m'_k$ is odd, the coefficient $Q_{\alpha_1 \cdots \alpha_{m'_k} \, , \, \beta_1 \cdots \beta_{n'_k}}$ is given by~(\ref{rel_m_odd})  for any $K \leqslant k \leqslant \mathbb{E} \left[ \left( \mathcal{N} + 1 \right) / 2 \right] - 1$.  
Accordingly, all the coefficients with an odd number of indices on the left of the comma can be substituted by using~(\ref{rel_m_odd}).

In order to formulate the problem, we define the partial sums $S_k$ by
\begin{eqnarray}
S_k \equiv \sum_{j=1}^{k} Q_{\bm{\alpha} \, , \, \bm{\beta}}^{(j)}
\label{S_k_def}
\end{eqnarray}
where $k$ is any integer such that $1 \leqslant k \leqslant n_K$. Note that for $k=n_K$ the sum $S_{n_K}$ is nothing but the right-hand side of~\eref{rel_m_K}. The idea is thus to show how the set of all the relations~\eref{rel_m_odd} for all the possible values of $m>m_K$ leads to $S_{n_K}=0$, i.e. precisely the result~\eref{rel_m_K}, hence theorem~\ref{theorem1}. This means that no information is gained from the relations obtained from~\eref{gen_rel_response_coef} for even indices $m$, and that only the relations corresponding to odd indices $m$ are relevant.

Our strategy is as follows. First, we note that the partial sums~(\ref{S_k_def}) obey the following recurrence:
\begin{eqnarray}
S_{k+1} = S_k + Q_{\bm{\alpha} \, , \, \bm{\beta}}^{(k+1)} \, .
\label{rec_S_k}
\end{eqnarray}
Using twice this recurrence relation we can simplify the partial sums of even indices $k$ by means of the relations~\eref{rel_m_odd}. Then we establish by induction that

\begin{proposition*}
For an arbitrary integer $l$ such that $1 \leqslant l \leqslant \mathbb{E} \left( \mathcal{N} / 2 \right) - K$, the partial sums~\eref{S_k_def} of even indices $k=2l$ can be expressed as
\begin{eqnarray}
\mathrm{H}_l : \hspace{1cm} S_{2l} = \sum_{j=2l+1}^{n_K} \gamma_{j}^{(2l)} Q_{\bm{\alpha} \, , \, \bm{\beta}}^{(j)}
\label{S_2l_ind_hyp_H_l}
\end{eqnarray}
where the quantities $\gamma_{j}^{(2l)}$ satisfy the recurrence
\begin{eqnarray}
\gamma_{j}^{(2l)} = - \frac{1}{2} \left( \gamma_{2l-1}^{(2l-2)} + 1 \right) \binom{j}{2l-1} + \gamma_{j}^{(2l-2)}
\label{gamma_j_2l_ind_hyp_H_l}
\end{eqnarray}
with $\gamma_j^{(0)}=0$ and the binomial coefficients
\begin{eqnarray}
\binom{p}{q} \equiv \frac{p!}{q! \, (p-q)!} \, .
\label{binom_coef_def}
\end{eqnarray}
\end{proposition*}

The above proposition is proved in several steps. In subsection~\ref{S_2_4_subsec}, we initiate the recurrence by explicitly calculating the partial sum $S_2$. Subsection~\ref{ind_proof_subsec} is then devoted to the proof of the recurrence with three lemmas that are crucial to show by induction that the hypothesis $\mathrm{H}_l$ implies $\mathrm{H}_{l+1}$ for any $l$. Finally, the expression~\eref{S_2l_ind_hyp_H_l} of the so-obtained partial sum $S_{2l}$ is used in subsection~\ref{S_n_K_subsec} to prove theorem~\ref{theorem1}, namely that $S_{n_K}=0$. We then readily prove the corollary.


\subsection{Proof of \texorpdfstring{${\rm H}_1$}{H1}}
\label{S_2_4_subsec}

In order to initiate the recurrence, we first explicitly compute the partial sum $S_2$ that, in view of its definition~\eref{S_k_def}, is given by
\begin{eqnarray}
S_2 = Q_{\bm{\alpha} \, , \, \bm{\beta}}^{(1)} + Q_{\bm{\alpha} \, , \, \bm{\beta}}^{(2)} 
\label{S_2_def}
\end{eqnarray}
where we have
\begin{eqnarray}
Q_{\bm{\alpha} \, , \, \bm{\beta}}^{(1)} = \sum_{k_1} Q_{\bm{\alpha} \beta_{k_1} \, , \, (\boldsymbol{\cdot})}
\label{Q1}
\end{eqnarray}
 according to the expression~\eref{Q_j_expr_remind} of $Q^{(j)}$.

We begin by noting that the quantity $Q_{\bm{\alpha} \beta_{k_1} \, , \, (\boldsymbol{\cdot})}$ contains, in view of the notation~\eref{alpha_def}, $m_K+1 = 2K+1$ subscripts on the left of the comma. Therefore, we can use the relation~\eref{rel_m_odd} for $k=K$ (since $m'_K = 2K+1$ by definition) to rewrite $Q_{\bm{\alpha} \beta_{k_1} \, , \, (\boldsymbol{\cdot})}$ and we get
\begin{eqnarray}
Q_{\bm{\alpha} \beta_{k_1} \, , \, (\boldsymbol{\cdot})} = - \frac{1}{2} \sum_{j=1}^{n'_K} \sum_{k'_{1} \atop k'_{1} \neq k_{1}} \sum_{k'_{2} > k'_1 \atop k'_{2} \neq k_{1}} \cdots \sum_{k'_{j} > k'_{j-1} \atop k'_{j} \neq k_{1}} Q_{\bm{\alpha} \beta_{k_1} \beta_{k'_1} \cdots \beta_{k'_j} \, (\boldsymbol{\cdot})}
\label{L_al_beta_k1}
\end{eqnarray}
because of the definition~\eref{Q_j_expr_remind} for $Q^{(j)}$ and the notation~\eref{notation_sums}.

We emphasize that the sums over $k'_1, \ldots , k'_j$ in the right-hand side of~\eref{L_al_beta_k1} must indeed run from 1 to $n_K$, and not $n'_K$ as the relation~\eref{rel_m_odd} might suggest at first sight. By construction, the quantity $Q_{\bm{\alpha}\beta_{k_1} \, , \, (\boldsymbol{\cdot})}$ contains, on the right of the comma, subscripts $\beta_k$ where $k$ can take any value between 1 and $n_K$, under the condition that $k \neq k_1$. Therefore, we must indeed have in~\eref{L_al_beta_k1} sums over $k'_1, \ldots , k'_j$ from 1 to $n_K$ (and \textit{not} $n'_K$), under the additional conditions that $k'_1 \neq k_1$, \ldots , $k'_j \neq k_1$.

Substituting the expression~\eref{L_al_beta_k1} of $Q_{\bm{\alpha}\beta_{k_1} \, , \, (\boldsymbol{\cdot})}$ into~\eref{Q1} yields
\begin{eqnarray}
Q_{\bm{\alpha} \, , \, \bm{\beta}}^{(1)} = - \frac{1}{2} \sum_{j=1}^{n'_K} S_{2 \, , \, j}
\label{S_2_temp_expr}
\end{eqnarray}
where the quantity $S_{2 \, , \, j}$ is defined by
\begin{eqnarray}
S_{2 \, , \, j} \equiv \sum_{k_1} \sum_{k'_{1} \atop k'_{1} \neq k_{1}} \sum_{k'_{2} > k'_1 \atop k'_{2} \neq k_{1}} \cdots \sum_{k'_{j} > k'_{j-1} \atop k'_{j} \neq k_{1}} Q_{\bm{\alpha}\beta_{k_1} \beta_{k'_1} \cdots \beta_{k'_j} \, , \, (\boldsymbol{\cdot})}
\label{S_2_j_def}
\end{eqnarray}
for any $1 \leqslant j \leqslant n'_K$. We now identify in $S_{2 \, , \, j}$ the coefficient of the term $Q_{\bm{\alpha} \beta_{h_1} \cdots \beta_{h_{j+1}} \, , \, (\boldsymbol{\cdot})}$ with $h_1 < \ldots < h_{j+1}$.

Note that the latter condition allows to define a set $\mathscr{S}_{j+1} \equiv \{ h_1,\ldots,h_{j+1} \}$ of $j+1$ elements. We then need to identify all \textit{samples} (which are ordered collections of elements of a set) $\{ k_1, k'_1,\ldots,k'_j \}$ of $\mathscr{S}_{j+1}$ that are compatible with the constraints $k'_1 , \ldots , k'_j \neq k_1$ \textit{and} $k'_1 < \ldots < k'_j$. This is thus equivalent to determining the total number of elements of the set $\mathscr{S}_{j+1}$, for fixing a value of $k_1$ in the set $\mathscr{S}_{j+1}$ unambiguously selects a unique subset $\{ k'_1,\ldots,k'_j \}$ of $\mathscr{S}_{j+1}$ because of the constraints. Therefore, we readily see that $S_{2 \, , \, j}$ contains exactly $j+1$ terms $Q_{\bm{\alpha} \beta_{h_1} \cdots \beta_{h_{j+1}} \, , \, (\boldsymbol{\cdot})}$, with $h_1 < \ldots < h_{j+1}$, so that
\begin{eqnarray}
S_{2 \, , \, j} = \left( j+1 \right) \sum_{k_1} \sum_{k_{2} > k_1} \cdots \sum_{k_{j+1} > k_{j}} Q_{\bm{\alpha}\beta_{k_1} \cdots \beta_{k_{j+1}} \, , \, (\boldsymbol{\cdot})} = (j+1) \, Q_{\bm{\alpha}\, , \, \bm{\beta}}^{(j+1)}
\label{S_2_j_expr}
\end{eqnarray}
by recognizing the expression~\eref{Q_j_expr_remind} of $Q^{(j)}$.  Now, we substitute~\eref{S_2_j_expr} into~\eref{S_2_temp_expr}, then \eref{S_2_temp_expr} into~(\ref{S_2_def}), make the change of index $j' = j+1$, and note that $n'_K+1 = n_K$ in view of~\eref{m+n_K_def} and~(\ref{prime_indices}). As a result, we find
\begin{eqnarray}
S_2 = \sum_{j=3}^{n_K} \gamma_{j}^{(2)} Q_{\bm{\alpha} \, , \, \bm{\beta}}^{(j)} 
\label{S_2_final_expr}
\end{eqnarray}
where the coefficients $\gamma_{j}^{(2)}$ are defined by
\begin{eqnarray}
\gamma_{j}^{(2)} \equiv - \frac{j}{2} 
\label{gamma_j_2_def}
\end{eqnarray}
for any integer $j \geqslant 3$.  Therefore, the partial sum $S_2$, as given by~\eref{S_2_final_expr}, is precisely of the form~\eref{S_2l_ind_hyp_H_l}, while the coefficients $\gamma_{j}^{(2)}$ obtained in~\eref{gamma_j_2_def} satisfy the relation~(\ref{gamma_j_2l_ind_hyp_H_l}) because of $\gamma^{(0)}_j=0$. This shows that the hypothesis $\mathrm{H}_1$ is indeed true.

Now that our induction has been initialized, we must show that the induction hypothesis $\mathrm{H}_l$ remains true for any integer $l$ such that $1 < l \leqslant \mathbb{E} \left( \mathcal{N} / 2 \right) - K$, which we prove in next subsection~\ref{ind_proof_subsec}.


\subsection{Proof of \texorpdfstring{${\rm H}_{l+1}$ from ${\rm H}_l$}{the induction}}
\label{ind_proof_subsec}

Having initialized our induction in section~\ref{S_2_4_subsec}, we now proceed to show that the hypothesis~\eref{S_2l_ind_hyp_H_l}-\eref{gamma_j_2l_ind_hyp_H_l} is indeed true. That is, we assume our hypothesis $\mathrm{H}_l$ to be true for some integer $l$ such that $1 \leqslant l < \mathbb{E} \left( \mathcal{N} / 2 \right) -K$.  Consequently, the partial sum $S_{2l}$ is given by~(\ref{S_2l_ind_hyp_H_l}). Hence we must show that $\mathrm{H}_{l+1}$ remains true.

The first step is to write down the partial sum $S_{2l+2}$ from the definition~\eref{S_k_def}, and we have
\begin{eqnarray}
S_{2l+2} = S_{2l} + Q_{\bm{\alpha} \, , \, \bm{\beta}}^{(2l+1)} + Q_{\bm{\alpha} \, , \, \bm{\beta}}^{(2l+2)} \, .
\label{S_l_plus_1_def}
\end{eqnarray}
We now substitute our induction hypothesis~\eref{S_2l_ind_hyp_H_l} into~\eref{S_l_plus_1_def} and gather the terms $Q_{\bm{\alpha} \, , \, \bm{\beta}}^{(2l+1)}$ to obtain
\begin{eqnarray}
S_{2l+2} = \left( \gamma_{2l+1}^{(2l)} + 1 \right) Q_{\bm{\alpha} \, , \, \bm{\beta}}^{(2l+1)} + Q_{\bm{\alpha} \, , \, \bm{\beta}}^{(2l+2)} + \sum_{j=2l+2}^{n_K} \gamma_{j}^{(2l)} Q_{\bm{\alpha} \, , \, \bm{\beta}}^{(j)} \, .
\label{S_l_plus_1_full_expr}
\end{eqnarray}
Now, we get
\begin{eqnarray}
Q_{\bm{\alpha} \, , \, \bm{\beta}}^{(2l+1)}= \sum_{k_1} \sum_{k_{2} > k_1} \cdots \sum_{k_{2l+1} > k_{2l}} Q_{\bm{\alpha} \beta_{k_1} \cdots \beta_{k_{2l+1}},\, (\boldsymbol{\cdot})}
\label{Q_2l+1}
\end{eqnarray}
by using the expression~\eref{Q_j_expr_remind} for $Q^{(j)}$.
Note that each term of the sum in~\eref{Q_2l+1} possesses $m_K+2l+1 = 2(K+l)+1$ subscripts on the left of the comma [remembering the definition~\eref{alpha_def} of $\bm{\alpha}$]. Therefore, we can use the relation~\eref{rel_m_odd} for $k=K+l$ [since $m'_{K+l} = 2(K+l)+1$ by definition] to rewrite $Q_{\bm{\alpha} \beta_{k_1} \cdots \beta_{k_{2l+1}} \, , \, (\boldsymbol{\cdot})}$. Moreover using the definition~\eref{Q_j_expr_remind} and the notation~\eref{notation_sums}, we obtain
\begin{eqnarray}
\fl Q_{\bm{\alpha} \beta_{k_1} \cdots \beta_{k_{2l+1}} \, , \, (\boldsymbol{\cdot})} = - \frac{1}{2} \sum_{j=1}^{n'_{K+l}} \sum_{k'_{1} \atop k'_{1} \neq k_{i}} \sum_{k'_{2} > k'_1 \atop k'_{2} \neq k_{i}} \cdots \sum_{k'_{j} > k'_{j-1} \atop k'_{j} \neq k_{i}} Q_{\bm{\alpha} \beta_{k_1} \cdots \beta_{k_{2l+1}} \beta_{k'_1} \cdots \beta_{k'_j} \, , \, (\boldsymbol{\cdot})} 
\label{Q_al_beta_k1_k2l_plus_1}
\end{eqnarray}
where $k' \neq k_i$ merely means $k' \neq k_1, \ldots , k_{2l+1}$. Here again, we emphasize that the sums over $k'_1$, \ldots , $k'_j$ in the right-hand side of~\eref{Q_al_beta_k1_k2l_plus_1} must run from 1 to $n_K$. Indeed, by construction the quantity $Q_{\bm{\alpha} \beta_{k_1} \cdots \beta_{k_{2l+1}} \, , \, (\boldsymbol{\cdot})}$ here contains, on the right of the comma, subscripts $\beta_k$ where $k$ can take any value between 1 and $n_K$, under the condition that $k \neq k_1, \ldots , k_{2l+1}$.

Therefore, the quantity~(\ref{Q_2l+1}) can be written as
\begin{eqnarray}
Q_{\bm{\alpha} \, , \, \bm{\beta}}^{(2l+1)}= -\frac{1}{2} \sum_{j=1}^{n'_{K+l}} S_{2l+2,j}
\label{Q_2l+1_S_2l+2,j}
\end{eqnarray}
in terms of 
\begin{eqnarray}
\fl S_{2l+2 \, , \, j} \equiv \sum_{k_1} \sum_{k_{2} > k_1} \cdots \sum_{k_{2l+1} > k_{2l}} \sum_{k'_{1} \atop k'_{1} \neq k_{i}} \sum_{k'_{2} > k'_1 \atop k'_{2} \neq k_{i}} \cdots \sum_{k'_{j} > k'_{j-1} \atop k'_{j} \neq k_{i}} Q_{\bm{\alpha} \beta_{k_1} \cdots \beta_{k_{2l+1}} \beta_{k'_1} \cdots \beta_{k'_j} \, , \, (\boldsymbol{\cdot})}
\label{S_l_plus_1_j_def}
\end{eqnarray}
for any $1 \leqslant j \leqslant n'_{K+l}$. We now identify in $S_{2l+2 \, , \, j}$ the coefficient of the term $Q_{\bm{\alpha} \beta_{h_1} \cdots \beta_{h_{j+2l+1}} \, , \, (\boldsymbol{\cdot})}$ with $h_1 < \ldots < h_{j+2l+1}$. To this end, we note that the latter condition defines a set $\mathscr{S}_{j+2l+1} \equiv \{ h_1,\ldots,h_{j+2l+1} \}$ of $j+2l+1$ elements. Our task then consists in counting all samples (which we recall are ordered collections of elements of a set) $\{ k_1,\ldots,k_{2l+1},k'_1,\ldots,k'_j \}$ of $\mathscr{S}_{j+2l+1}$ that are compatible with the constraints $k'_1 , \ldots , k'_j \neq k_1, \ldots , k_{2l+1}$, $k_1 < \ldots < k_{2l+1}$ \textit{and} $k'_1 < \ldots < k'_j$. This problem is thus equivalent to determining the total number of \textit{subsets} (which are unordered collections of elements of a set) $\{ k_1,\ldots,k_{2l+1} \}$ of the set $\mathscr{S}_{j+2l+1}$. Indeed, note in particular that fixing $\{ k_1,\ldots,k_{2l+1} \}$ to be a definite subset of $\mathscr{S}_{j+2l+1}$ unambiguously selects a unique subset $\{ k'_1,\ldots,k'_j \}$ of $\mathscr{S}_{j+2l+1}$. This number is known to be $\binom{j+2l+1}{2l+1}$ (see e.g. theorem~4.1 in~\cite{Rys}).  Consequently, $S_{2l+2 \, , \, j}$ contains exactly $\binom{j+2l+1}{2l+1}$ terms $Q_{\bm{\alpha} \beta_{h_1} \cdots \beta_{h_{j+2l+1}} \, , \, (\boldsymbol{\cdot})}$ with $h_1 < \ldots < h_{j+2l+1}$. Therefore, we see that~\eref{S_l_plus_1_j_def} reads
\begin{eqnarray}
S_{2l+2 \, , \, j} = \binom{j+2l+1}{2l+1} \sum_{k_1} \sum_{k_{2} > k_1} \cdots \sum_{k_{j+2l+1} > k_{j+2l}} Q_{\bm{\alpha} \beta_{k_1} \cdots \beta_{k_{j+2l+1}} \, , \, (\boldsymbol{\cdot})}
\label{S_l_plus_1_j_temp_expr}
\end{eqnarray}
whereupon
\begin{eqnarray}
S_{2l+2 \, , \, j} = \binom{j+2l+1}{2l+1} \, Q_{\bm{\alpha} \, , \, \bm{\beta}}^{(j+2l+1)} 
\label{S_l_plus_1_j_expr}
\end{eqnarray}
after recognizing in~\eref{S_l_plus_1_j_temp_expr} the quantity $Q^{(j)}$ as given by~\eref{Q_j_expr_remind}.

Now, we substitute this expression of $S_{2l+2 \, , \, j}$ into~\eref{Q_2l+1_S_2l+2,j}, which is itself replaced into~(\ref{S_l_plus_1_full_expr}).
We make the change of index $j' = j+2l+1$, and note that $n'_{K+l}+2l+1 = n_K$ in view of~\eref{m+n_K_def} and~(\ref{prime_indices}). As a result, we see that the sum~\eref{S_l_plus_1_full_expr} becomes
\begin{eqnarray}
\fl S_{2l+2} = - \frac{1}{2} \left( \gamma_{2l+1}^{(2l)} + 1 \right) \sum_{j=2l+2}^{n_K} \binom{j}{2l+1} \, Q_{\bm{\alpha} \, , \, \bm{\beta}}^{(j)} + Q_{\bm{\alpha} \, , \, \bm{\beta}}^{(2l+2)} + \sum_{j=2l+2}^{n_K} \gamma_{j}^{(2l)} Q_{\bm{\alpha} \, , \, \bm{\beta}}^{(j)} \, .
\label{S_l_plus_1_temp_expr}
\end{eqnarray}
Therefore, we find that the partial sum $S_{2l+2}$ can be written as
\begin{eqnarray}
S_{2l+2} = \sum_{j=2l+3}^{n_K} \gamma_{j}^{(2l+2)} Q_{\bm{\alpha} \, , \, \bm{\beta}}^{(j)} + \left( \gamma_{2l+2}^{(2l+2)} + 1 \right) Q_{\bm{\alpha} \, , \, \bm{\beta}}^{(2l+2)}
\label{S_l_plus_1_temp_expr_gam}
\end{eqnarray}
with the coefficients
\begin{eqnarray}
\gamma_{j}^{(2l+2)} = - \frac{1}{2} \left( \gamma_{2l+1}^{(2l)} + 1 \right) \binom{j}{2l+1} + \gamma_{j}^{(2l)} 
\label{gamma_j_l_plus_1_def}
\end{eqnarray}
in agreement with the relation~(\ref{gamma_j_2l_ind_hyp_H_l}) for the induction hypothesis ${\rm H}_{l+1}$ instead of ${\rm H}_l$.

Comparing~(\ref{S_l_plus_1_temp_expr_gam}) with~(\ref{S_2l_ind_hyp_H_l}), we see that, in order to complete the proof of the recurrence, we need to show that the coefficient~\eref{gamma_j_l_plus_1_def} for $j=2l+2$ is merely equal to  $\gamma_{2l+2}^{(2l+2)}=-1$ independently of $l$.  This result is established with lemma \ref{lemma_gamma_2l_constant} below, which requires the two preliminary lemmas~\ref{lemma_gamma_from_Gamma} and~\ref{lemma_Gamma_from_Euler_coef}, as we now discuss.

Our first lemma expresses the coefficient $\gamma_{j}^{(2l)}$ in terms of the coefficients $\gamma_{k}^{(2)}$, as follows:

\begin{lemma}
Let~$l$ and~$j$ be two integers such that $l \geqslant 2$ and $j \geqslant 2l$. Let $\gamma_{j}^{(2l)}$ be a quantity that satisfies the recurrence~\eref{gamma_j_2l_ind_hyp_H_l} with $\gamma_{j}^{(2)} = -j/2$. Then the coefficient $\gamma_{j}^{(2l)}$ can be written in the form
\begin{eqnarray}
\gamma_{j}^{(2l)} = \gamma_{j}^{(2)} + \sum_{k=1}^{l-1} \Gamma_{j,k}^{(2l)} \left( \gamma_{2l-2k+1}^{(2)} + 1 \right)
\label{gamma_j_2l_from_Gamma_j_k_2l_def}
\end{eqnarray}
where the coefficient $\Gamma_{j,k}^{(2l)}$ must satisfy the recurrence
\begin{eqnarray}
\Gamma_{j,k}^{(2l)} = - \frac{1}{2} \sum_{i=1}^{k-1} \binom{2l-2i+1}{2l-2k+1} \Gamma_{j,i}^{(2l)} - \frac{1}{2} \binom{j}{2l-2k+1}
\label{Gamma_j_k_2l_def}
\end{eqnarray}
and where it must be understood that for $k=1$ we merely have
\begin{eqnarray}
\Gamma_{j,1}^{(2l)} = - \frac{1}{2} \binom{j}{2l-1} \, .
\label{Gamma_j_1_2l_def}
\end{eqnarray}
\label{lemma_gamma_from_Gamma}
\end{lemma}

The form~\eref{gamma_j_2l_from_Gamma_j_k_2l_def} of the quantity $\gamma_{j}^{(2l)}$ presents an important advantage over the recurrence~\eref{gamma_j_2l_ind_hyp_H_l}. Indeed, it explicitly expresses $\gamma_{j}^{(2l)}$, for any integers $l \geqslant 2$ and $j \geqslant 2l$, in terms of known values $\gamma_{i}^{(2)} = -i/2$. The counterpart is that the coefficients $\Gamma_{j,k}^{(2l)}$ involved in~\eref{gamma_j_2l_from_Gamma_j_k_2l_def} are themselves defined in the recursive form~\eref{Gamma_j_k_2l_def}. However, the latter turn out to be directly related to the coefficients of the Euler polynomials, as will be clear from lemma~\ref{lemma_Gamma_from_Euler_coef}.

\begin{proof}
We proceed by induction over the index $l$, and show that the induction hypothesis defined by~\eref{gamma_j_2l_from_Gamma_j_k_2l_def}
with the coefficients $\Gamma_{j,k}^{(2l)}$ satisfying~\eref{Gamma_j_k_2l_def} is true for any integer $l \geqslant 2$.

We first consider the case $l=2$, and explicitly construct the coefficient $\gamma_{j}^{(4)}$ from its definition~\eref{gamma_j_2l_ind_hyp_H_l}. Comparing with~\eref{gamma_j_2l_from_Gamma_j_k_2l_def}, we find that
\begin{eqnarray*}
\Gamma_{j,1}^{(4)} \equiv - \frac{1}{2} \binom{j}{3}
\end{eqnarray*}
in agreement with~\eref{Gamma_j_1_2l_def}, hence the hypothesis is true for $l=2$.

We now assume our hypothesis to be true for some integer $l \geqslant 2$. Then we compute $\gamma_{j}^{(2l+2)}$ from its definition~\eref{gamma_j_2l_ind_hyp_H_l} and replace therein $\gamma_{j}^{(2l)}$ by its expression~\eref{gamma_j_2l_from_Gamma_j_k_2l_def}.  Introducing
\begin{eqnarray}
\Gamma_{j,1}^{(2l+2)} \equiv - \frac{1}{2} \binom{j}{2l+1}
\label{Gamma_j_1_ind}
\end{eqnarray}
in agreement with the definition~\eref{Gamma_j_1_2l_def}, we hence get
\begin{eqnarray}
\fl \gamma_{j}^{(2l+2)} = \gamma_{j}^{(2)} + \Gamma_{j,1}^{(2l+2)} \left( \gamma_{2l+1}^{(2)} + 1 \right) + \sum_{k=1}^{l-1} \left( \Gamma_{j,k}^{(2l)} + \Gamma_{j,1}^{(2l+2)} \Gamma_{2l+1,k}^{(2l)} \right) \left( \gamma_{2l-2k+1}^{(2)} + 1 \right) .
\label{gamma_from_Gamma_ind_temp}
\end{eqnarray}
Now, we define the coefficients $\Gamma_{j,k}^{(2l+2)}$ by
\begin{eqnarray}
\Gamma_{j,k}^{(2l+2)} \equiv \Gamma_{j,k-1}^{(2l)} + \Gamma_{j,1}^{(2l+2)} \Gamma_{2l+1,k-1}^{(2l)}
\label{Gamma_ind_def}
\end{eqnarray}
with $k \geqslant 2$. Substituting this definition into~\eref{gamma_from_Gamma_ind_temp} and making the change of index $k' = k+1$ readily yields
\begin{eqnarray}
\gamma_{j}^{(2l+2)} = \gamma_{j}^{(2)} + \sum_{k=1}^{l} \Gamma_{j,k}^{(2l+2)} \left( \gamma_{2l-2k+3}^{(2)} + 1 \right)  .
\label{gamma_from_Gamma_final_expr}
\end{eqnarray}
Therefore, the coefficient $\gamma_{j}^{(2l+2)}$ is already written in the required form for the induction hypothesis to be true at next step $l+1$.

Finally, we must show that the quantities $\Gamma_{j,k}^{(2l+2)}$ defined by~\eref{Gamma_ind_def} are recursively related precisely in the way~\eref{Gamma_j_k_2l_def}. In~\eref{Gamma_ind_def} we express $\Gamma_{j,k-1}^{(2l)}$ and $\Gamma_{2l+1,k-1}^{(2l)}$ by means of the induction hypothesis~\eref{Gamma_j_k_2l_def}. The result is then simplified by using the definition~\eref{Gamma_ind_def} for $\Gamma_{j,i+1}^{(2l+2)}$. A change of index $i' = i+1$ hence readily yields
\begin{eqnarray}
\Gamma_{j,k}^{(2l+2)} = - \frac{1}{2} \sum_{i=1}^{k-1} \binom{2l-2i+3}{2l-2k+3} \Gamma_{j,i}^{(2l+2)} - \frac{1}{2} \binom{j}{2l-2k+3} \, .
\label{Gamma_ind_final_expr}
\end{eqnarray}

Therefore, we showed that the coefficient $\gamma_{j}^{(2l+2)}$ can be expressed through~\eref{gamma_from_Gamma_final_expr} in terms of quantities $\Gamma_{j,k}^{(2l+2)}$ that must satisfy the recursive relation~\eref{Gamma_ind_final_expr}. This clearly shows that the induction hypothesis~\eref{gamma_j_2l_from_Gamma_j_k_2l_def}-\eref{Gamma_j_1_2l_def} holds for any integer $l \geqslant 2$.
\end{proof}

An interesting consequence of lemma \ref{lemma_gamma_from_Gamma} is that the quantities $\Gamma_{j,k}^{(2l)}$ given by~\eref{Gamma_j_k_2l_def} are related to the Euler polynomials. Before we explicitly get this result in lemma~\ref{lemma_Gamma_from_Euler_coef}, we first recall a few well-established results regarding these particular polynomials.

The Euler polynomial $E_n(x)$, with $n \geqslant 0$ an integer and $x$ a real number, are defined with their generating function \cite{GradRyz, AbrSteg} by
\begin{eqnarray}
\frac{2\, {\rm e}^{xt}}{{\rm e}^t+1} = \sum_{n=0}^{\infty} E_n(x) \, \frac{t^n}{n!} \, .
\label{Euler_poly_def}
\end{eqnarray}
They can be written as
\begin{eqnarray}
E_n (x) = \sum_{i=0}^{n} e_{i}^{(n)} x^{i}
\label{Euler_coef_def}
\end{eqnarray}
in terms of the coefficients $e_{i}^{(n)}$ such that $e_n^{(n)}=1$.
Setting $x=0$ into~\eref{Euler_coef_def} and using a well-known property of Euler polynomials \cite{AbrSteg} readily shows that the constant term $e_{0}^{(n)}$ of $E_n (x)$ is merely given by
\begin{eqnarray}
e_{0}^{(n)} = E_n (0) = -E_n(1)
\label{constant_term_Euler}
\end{eqnarray}
for any integer $n \geqslant 0$.  Moreover, it is known that
\begin{eqnarray}
e_0^{(0)}=E_0(0)=1 \qquad \mbox{and}\qquad e_{0}^{(2k)} = E_{2k}(0)=0
\label{e_0_even_vanish}
\end{eqnarray}
for any integer $k \geqslant 1$ \cite{GradRyz, AbrSteg}.  Therefore, except for $e_0^{(0)}$, the constant terms of Euler polynomials are not identically zero only for odd indices.

Alternatively, the constant terms $e_0^{(n)}$ can be related to coefficients $e_i^{(n)}$ of the Euler polynomials. Indeed, the latter are known to satisfy the identity \cite{AbrSteg}
\begin{eqnarray}
E_n (h+x) = \sum_{i=0}^{n} \binom{n}{i} E_i (h) \, x^{n-i}
\label{E_k_h_plus_x}
\end{eqnarray}
for any integer $n \geqslant 0$. Setting $h=0$ into~\eref{E_k_h_plus_x} hence yields, in view of~\eref{Euler_coef_def} and~\eref{constant_term_Euler},
\begin{eqnarray}
\sum_{i=0}^{n} e_{i}^{(n)} x^{i} = \sum_{i=0}^{n} \binom{n}{i} e_{0}^{(i)} x^{n-i} \, .
\label{Euler_poly_id}
\end{eqnarray}
Identifying the terms corresponding to $i=n-p$ (resp. $i=p$) in the left- (resp. right-) hand side of~\eref{Euler_poly_id} hence shows that 
\begin{eqnarray}
e_{n-p}^{(n)} = e_{0}^{(p)} \binom{n}{p}
\label{e_from_e_0}
\end{eqnarray}
for any integers $n$ and $p$ such that $n \geqslant 0$ and $0 \leqslant p \leqslant n$. It is thus worth noting that this relation readily implies, in view of~\eref{e_0_even_vanish}, that
\begin{eqnarray}
e_{n-2q}^{(n)} = 0
\label{e_k_min_2p_vanish}
\end{eqnarray}
for any integers $n$ and $q$ such that $n \geqslant 2$ and $1 \leqslant q \leqslant \mathbb{E} (n/2)$.

Now, the above well-known properties can be adequately used to rewrite the quantity $\Gamma_{2l,k}^{(2l)}$, which occurs in the expression~\eref{gamma_j_2l_from_Gamma_j_k_2l_def} of $\gamma_{2l}^{(2l)}$. Indeed, we show that $\Gamma_{2l,k}^{(2l)}$ can be expressed in terms of the coefficients $e_i^{(n)}$ of Euler polynomials as

\begin{lemma}
Let $l$ and $k$ be two integers such that $l \geqslant 2$ and $1 \leqslant k \leqslant l-1$. Let $\Gamma_{j,k}^{(2l)}$ be a quantity that satisfies~\eref{Gamma_j_k_2l_def}-\eref{Gamma_j_1_2l_def} for any integer $j \geqslant 2l$. Then the particular coefficient $\Gamma_{2l,k}^{(2l)}$ is related to the coefficients $e_i$ of Euler polynomials through
\begin{eqnarray}
\Gamma_{2l,k}^{(2l)} = e_{2l-2k+1}^{(2l)} = e_{0}^{(2k-1)} \binom{2l}{2k-1}  .
\label{Gamma_2l_from_Euler_coef}
\end{eqnarray}
\label{lemma_Gamma_from_Euler_coef}
\end{lemma}

The gain offered by the form~\eref{Gamma_2l_from_Euler_coef} of the quantity $\Gamma_{2l,k}^{(2l)}$ as compared to the original definition~\eref{Gamma_j_k_2l_def} should be clear. First, the latter expresses $\Gamma_{2l,k}^{(2l)}$ recursively in terms of all the $\Gamma_{2l,i}^{(2l)}$, with $1 \leqslant i \leqslant k-1$, what is no longer the case in~\eref{Gamma_2l_from_Euler_coef}. Second, this direct relation between $\Gamma_{2l,k}^{(2l)}$ and $e_i^{(n)}$ allows us to manipulate the quantities $\Gamma_{2l,k}^{(2l)}$, and thus in view of~\eref{gamma_j_2l_from_Gamma_j_k_2l_def} the coefficient $\gamma_{2l}^{(2l)}$ itself, by means of known identities regarding the coefficients of Euler polynomials. This latter point then appears to be of crucial importance to prove the main result of this subsection, namely lemma~\ref{lemma_gamma_2l_constant}.

\begin{proof}
Let $l \geqslant 2$ be a fixed integer. Setting $j=2l$ into the definition~\eref{Gamma_j_k_2l_def}, we get for the coefficient $\Gamma_{2l,k}^{(2l)}$
\begin{eqnarray}
\Gamma_{2l,k}^{(2l)} = - \frac{1}{2} \sum_{i=1}^{k-1} \binom{2l-2i+1}{2k-2i} \Gamma_{2l,i}^{(2l)} - \frac{1}{2} \binom{2l}{2k-1}
\label{Gamma_2l_k_2l_def}
\end{eqnarray}
for any integer $k$ such that $1 \leqslant k \leqslant l-1$.

We now proceed by induction over the index $k$ of $\Gamma_{2l,k}^{(2l)}$, and show that the induction hypothesis~\eref{Gamma_2l_from_Euler_coef} holds for any $1 \leqslant k \leqslant l-1$.

We first set $k=1$ into~\eref{Gamma_2l_k_2l_def}. Noting that $e_0^{(1)}=-1/2$ readily shows that \eref{Gamma_2l_from_Euler_coef} holds for $k=1$.

Now, we assume our hypothesis~\eref{Gamma_2l_from_Euler_coef} to be true for some integer $k$, $1 \leqslant k < l-1$.  We first compute the coefficient $\Gamma_{2l,k+1}^{(2l)}$ from its definition~\eref{Gamma_2l_k_2l_def} and make use of the induction hypothesis~\eref{Gamma_2l_from_Euler_coef} for $k=i$. Since
\begin{eqnarray*}
\binom{2l-2i+1}{2k-2i+2} \binom{2l}{2i-1} = \binom{2l}{2k+1} \binom{2k+1}{2i-1} ,
\end{eqnarray*}
we get, after the change of index $i'=i-1$, then making use of the identity~\eref{e_from_e_0}, and noting that $e_{2k+1}^{(2k+1)}=1$,
\begin{eqnarray}
\Gamma_{2l,k+1}^{(2l)} = - \frac{1}{2} \binom{2l}{2k+1} \left( e_{2k+1}^{(2k+1)} + \sum_{i=0}^{k-1} e_{2k-2i}^{(2k+1)} \right) .
\label{Gamma_2l_k_plus_1_2l_expr_temp}
\end{eqnarray}

Now, in view of the definition~\eref{Euler_coef_def} of the Euler polynomial $E_{2k+1} (x)$ and the property~\eref{e_k_min_2p_vanish} of the coefficients $e_{2k+1-2p}^{(2k+1)}$, it is clear that
\begin{eqnarray}
e_{2k+1}^{(2k+1)} + \sum_{i=0}^{k-1} e_{2k-2i}^{(2k+1)} = E_{2k+1} (1) - e_{0}^{(2k+1)} = - 2\, e_{0}^{(2k+1)}
\label{lin_comb_Euler_coef}
\end{eqnarray}
where we used~\eref{constant_term_Euler} to write the last equality. Combining~\eref{Gamma_2l_k_plus_1_2l_expr_temp} with~\eref{lin_comb_Euler_coef} readily yields~\eref{Gamma_2l_from_Euler_coef} for $k$ replaced by $k+1$. Therefore, the relation~\eref{Gamma_2l_from_Euler_coef} holds for any integer $k$ such that $1 \leqslant k \leqslant l-1$. 
\end{proof}

Finally, we use the above results, lemmas~\ref{lemma_gamma_from_Gamma} and~\ref{lemma_Gamma_from_Euler_coef}, to show that the coefficient $\gamma_{2l}^{(2l)}$ satisfies

\begin{lemma}
Let $l$ be an integer such that $l \geqslant 2$. Let $\gamma_{j}^{(2l)}$ be a quantity that satisfies~\eref{gamma_j_2l_ind_hyp_H_l} for any integer $j$ such that $j \geqslant 2l$. Then the particular coefficient $\gamma_{2l}^{(2l)}$ is constant and given by
\begin{eqnarray}
\gamma_{2l}^{(2l)} = - 1 \, .
\label{gamma_2l_2l_constant}
\end{eqnarray}
\label{lemma_gamma_2l_constant}
\end{lemma}

This lemma is the most important result of the present subsection, for it provides the end of our proof of the induction hypothesis ${\rm H}_l$ given by~\eref{S_2l_ind_hyp_H_l}. Indeed, to prove that the hypothesis $\mathrm{H}_{l+1}$ remains true requires that $\gamma_{2l+2}^{(2l+2)} + 1 = 0$, which is ensured by the above lemma. This cancels the term $Q_{\bm{\alpha} \, , \, \bm{\beta}}^{(2l+2)}$ in the expression~\eref{S_l_plus_1_temp_expr_gam} of the partial sum $S_{2l+2}$, which can then be written in the desired form.

\begin{proof}
Setting $j=2l$ into~\eref{gamma_j_2l_from_Gamma_j_k_2l_def} and remembering that $\gamma_{j}^{(2)} = -j/2$ for any integer $j \geqslant 2$, we get 
\begin{eqnarray}
\gamma_{2l}^{(2l)} = -l - \frac{1}{2} \sum_{k=1}^{l-1} \left( 2l-2k+1 \right) \Gamma_{2l,k}^{(2l)} + \sum_{k=1}^{l-1} \Gamma_{2l,k}^{(2l)} \, .
\label{gamma_2l_2l_Gamma}
\end{eqnarray}
Using~\eref{Gamma_2l_from_Euler_coef} of lemma~\ref{lemma_Gamma_from_Euler_coef} to rewrite $\Gamma_{2l,k}^{(2l)}$, and noting that
\begin{eqnarray*}
\left( 2l-2k+1 \right) \binom{2l}{2k-1} = 2l \binom{2l-1}{2k-1}  ,
\end{eqnarray*}
we obtain
\begin{eqnarray}
\gamma_{2l}^{(2l)} = -l - l \sum_{k=1}^{l-1} \binom{2l-1}{2k-1} e_0^{(2k-1)} + \sum_{k=1}^{l-1} \binom{2l}{2k-1} e_0^{(2k-1)} \, .
\label{gamma_2l_2l_Euler}
\end{eqnarray}
Now, the properties~\eref{constant_term_Euler}-\eref{e_0_even_vanish} and~\eref{e_from_e_0}-\eref{e_k_min_2p_vanish} of Euler polynomials~\eref{Euler_coef_def} at $x=1$ give
\begin{eqnarray}
\fl&&\sum_{k=1}^{l-1} \binom{2l-1}{2k-1} e_0^{(2k-1)} = E_{2l-1}(1)-1-e_0^{(2l-1)} =-1-2\, e_0^{(2l-1)}  \qquad\qquad\mbox{and} \label{Euler_sum_(2l-1)} \\
\fl&&\sum_{k=1}^{l-1} \binom{2l}{2k-1} e_0^{(2k-1)} = E_{2l}(1)-1-\binom{2l}{2l-1} e_0^{(2l-1)} =-1-2l\, e_0^{(2l-1)} \, . 
\label{Euler_sum_(2l)}
\end{eqnarray}
Substituting~\eref{Euler_sum_(2l-1)} and~\eref{Euler_sum_(2l)} back into~\eref{gamma_2l_2l_Euler}, readily shows the desired result~\eref{gamma_2l_2l_constant}.
\end{proof}

It is worth noting that as a direct by-product of the above results, lemmas~\ref{lemma_gamma_from_Gamma}-\ref{lemma_gamma_2l_constant}, we obtain a general identity satisfied by coefficients of Euler polynomials. Indeed, we recall that the latter are known to satisfy
\begin{eqnarray}
e_{2k}^{(2k+1)} = - \frac{2k+1}{2}
\label{second_highest_order_Euler_coef}
\end{eqnarray}
for any integer $k \geqslant 0$. Since we have by construction $\gamma_{j}^{(2)} \equiv -j/2$, for any $j \geqslant 2$, it is then clear from~\eref{second_highest_order_Euler_coef} that we have
\begin{eqnarray}
\gamma_{2l-2k+1}^{(2)} + 1 = e_{2l-2k-2}^{(2l-2k-1)}
\label{gamma_plus_1_from_Euler_coef}
\end{eqnarray}
for any integers $l$ and $k$ such that $l \geqslant 2$ and $1 \leqslant k \leqslant l-1$. Therefore, combining lemmas~\ref{lemma_gamma_from_Gamma}-\ref{lemma_gamma_2l_constant} with~\eref{gamma_plus_1_from_Euler_coef} yields the general identity
\begin{eqnarray}
\sum_{k=1}^{l-1} e_{2l-2k+1}^{(2l)} \, e_{2l-2k-2}^{(2l-2k-1)} = l-1
\label{gamma_2l_2l_from_Euler_coef}
\end{eqnarray}
for any $l \geqslant 2$. This provides a simple closed form expression of a particular linear combination of products of coefficients of Euler polynomials.

With the result of lemma~\ref{lemma_gamma_2l_constant} at hand, we are in position to show that the hypothesis $\mathrm{H}_l$, as given by~\eref{S_2l_ind_hyp_H_l}-\eref{gamma_j_2l_ind_hyp_H_l}, is true for any integer $l$ such that $1 \leqslant l \leqslant \mathbb{E} \left( \mathcal{N} / 2 \right) - K$. Indeed, it is precisely the latter lemma that allows us to write the partial sum $S_{2l+2}$ in the form~\eref{S_2l_ind_hyp_H_l}, as can be seen from the identity~\eref{gamma_2l_2l_constant} that readily yields
\begin{eqnarray}
\gamma_{2l+2}^{(2l+2)} + 1 = 0 \, .
\label{gamma_ind_identity}
\end{eqnarray}
We now substitute~\eref{gamma_ind_identity} into~\eref{S_l_plus_1_temp_expr_gam} to get
\begin{eqnarray}
S_{2l+2} = \sum_{j=2(l+1)+1}^{n_K} \gamma_{j}^{(2l+2)} Q_{\bm{\alpha} \, , \, \bm{\beta}}^{(j)} \, .
\label{S_l_plus_1_final_expr}
\end{eqnarray}
Combining the expression~\eref{S_l_plus_1_final_expr} of $S_{2l+2}$ with the definition~\eref{gamma_j_l_plus_1_def} of the coefficients $\gamma_{j}^{(2l+2)}$ readily shows that the hypothesis $\mathrm{H}_{l+1}$ is indeed true.

Hence we showed (i) that the hypothesis $\mathrm{H}_l$, as defined by~\eref{S_2l_ind_hyp_H_l}-\eref{gamma_j_2l_ind_hyp_H_l}, is true for $l=1$ (this was done in subsection~\ref{S_2_4_subsec}), and (ii) that $\mathrm{H}_{l+1}$ must be true if $\mathrm{H}_l$ is. Therefore, the hypothesis $\mathrm{H}_l$ is indeed true for any integer $l$ such that $1 \leqslant l \leqslant \mathbb{E} \left( \mathcal{N} / 2 \right) - K$, and the partial sum $S_{2l}$ is given by~(\ref{S_2l_ind_hyp_H_l}). We have thus proved the proposition.

The validity of the expression~\eref{S_2l_ind_hyp_H_l} of the partial sums $S_{2l}$ now enables us to tackle our initial problem, namely to show that the total sum $S_{n_K}$ vanishes, and thus that theorem~\ref{theorem1} is indeed true. Finally, this allows us to prove the corollary.


\subsection{Proof of theorem \texorpdfstring{\ref{theorem1}}{} and corollary}\label{S_n_K_subsec}

The above analysis shows that the partial sum $S_{2l}$ can be written in the form~\eref{S_2l_ind_hyp_H_l} for any integer $l$ such that $1 \leqslant l \leqslant \mathbb{E} \left( \mathcal{N} / 2 \right) - K$. This result can now be adequately used to compute the total sum $S_{n_K}$. We do this by distinguishing the two cases of an even and of an odd total number $\mathcal{N}$ of subscripts.

We first suppose that $\mathcal{N}$ is even, say of the form $\mathcal{N}=2N$, $N \geqslant 1$, so that $\mathbb{E} \left( \mathcal{N} / 2 \right) - K = N-K$ and, from the definition~\eref{m+n_K_def}, $n_K = 2N-2K$. Hence setting $l=N-K$ into~\eref{S_2l_ind_hyp_H_l} yields
\begin{eqnarray*}
S_{n_K} = S_{2N-2K} = \sum_{j=n_K+1}^{n_K} \gamma_{j}^{(n_K)} Q_{\bm{\alpha} \, , \, \bm{\beta}}^{(j)} 
\end{eqnarray*}
that is clearly $S_{n_K} = 0$.

We now suppose that $\mathcal{N}$ is odd, say of the form $\mathcal{N}=2N+1$, $N \geqslant 0$. In this case we still have $\mathbb{E} \left( \mathcal{N} / 2 \right) - K = N-K$, but $n_K$ now reads $n_K = 2N-2K+1$. Setting again $l=N-K$ into~\eref{S_2l_ind_hyp_H_l} and using the recurrence~\eref{rec_S_k} yields
\begin{eqnarray}
S_{n_K} = S_{2N-2K} + Q_{\bm{\alpha} \, , \, \bm{\beta}}^{(n_K)} = \left( \gamma_{n_K}^{(2N-2K)} + 1 \right) Q_{\bm{\alpha} \, , \, \bm{\beta}}^{(n_K)}
\label{S_n_K_N_odd}
\end{eqnarray}
with, in view of the definition~\eref{Q_j_expr_remind} for $Q^{(j)}$,
\begin{eqnarray}
Q_{\bm{\alpha} \, , \, \bm{\beta}}^{(n_K)} = \sum_{k_1} \sum_{k_{2} > k_{1}} \cdots \sum_{k_{n_K} > k_{n_K-1}} Q_{\bm{\alpha}\beta_{k_1} \cdots \beta_{k_{n_K}}} \left( \bm{0} \right)
\label{Q_n_K_expr}
\end{eqnarray}
where we emphasize that the elements of the sum are statistical cumulants (that indeed do not possess any subscript on the right of the comma). We recall [see~\eref{alpha_def} and~\eref{m+n_K_def}] that $\bm{\alpha}\equiv \alpha_1 \ldots \alpha_{m_K}$ and $m_K+n_K \equiv \mathcal{N}$, and that we consider an odd total number of subscripts, $\mathcal{N}=2N+1$. We can thus use the relation~\eref{rel_m_odd} for $k=N$ to rewrite the cumulant $Q_{\bm{\alpha}\beta_{k_1} \cdots \beta_{k_{n_K}}} \left( \bm{0} \right)$. Since in view of the definition~\eref{prime_indices} we have in this case $m'_N \equiv 2N+1=\mathcal{N}$ and $n'_N \equiv \mathcal{N}-m'_N=0$, it is then clear that the latter cumulants vanish identically,
\begin{eqnarray}
\fl Q_{\bm{\alpha}\beta_{k_1} \cdots \beta_{k_{n_K}},} = Q_{\bm{\alpha}\beta_{k_1} \cdots \beta_{k_{n_K}}} \left( \bm{0} \right) = 0 \qquad\mbox{for odd values of}\ \ \mathcal{N} = m_K+n_K\, .
\label{last_cumulants_vanish}
\end{eqnarray}
Combining~\eref{S_n_K_N_odd}-\eref{Q_n_K_expr} with~\eref{last_cumulants_vanish} readily yields, here again, $S_{n_K} = 0$.

Hence we see that the total sum $S_{n_K}$ vanishes both (i) for $\mathcal{N}$ even, and (ii) for $\mathcal{N}$ odd, so that we indeed have
\begin{eqnarray}
S_{n_K} = 0 
\label{S_n_K_zero}
\end{eqnarray}
for any integer $\mathcal{N} \geqslant 1$, which concludes the proof of theorem~\ref{theorem1}. This shows that, for a fixed $\mathcal{N}=m+n$, the relation obtained from~\eref{gen_rel_response_coef} for an arbitrary even index $m=m_K \equiv 2K$ can be deduced from the set of all relations~\eref{gen_rel_response_coef} for all odd indices $m>m_K$. In other words, the relations that are generated by the FR~\eref{fluct_rel} for a given arbitrary total number $\mathcal{N} = m+n \geqslant 1$ of subscripts are not all independent. Only the relations that are obtained for an odd index $m$ provide non-trivial information about the quantities $Q_{\alpha_1 \cdots \alpha_m \, , \, \beta_1 \cdots \beta_n}$.

The above theorem~\ref{theorem1} hence allows us to unambiguously identify, among all relations \eref{gen_rel_response_coef} generated by the FR~\eref{fluct_rel}, which of them are truly independent. Indeed, we can now establish that these independent relations are precisely the ones obtained from~\eref{gen_rel_response_coef} for odd values of the index $m$. This is a direct consequence of the structure of these relations. To see this, we again consider a given arbitrary total number $\mathcal{N} = m+n$ of subscripts, as well as an arbitrary odd integer $M$ such that $1 \leqslant M \leqslant \mathcal{N}$. The relation~\eref{rel_m_odd} for $m=M$ expresses the coefficient $Q_{\alpha_1 \cdots \alpha_M \, , \, \beta_1 \cdots \beta_{\mathcal{N}-M}}$ in terms of quantities that possess $M+1, \ldots , \mathcal{N}$ subscripts on the left of the comma. This coefficient cannot be expressed from any of the relations~\eref{rel_m_odd} for $m>M$ since it does not appear in the latter relations. Consequently, the relation~\eref{rel_m_odd} for $m=M$ cannot be deduced from the set of relations~\eref{rel_m_odd} for all $m>M$. This result holds for an arbitrary odd integer $M \leqslant \mathcal{N}$. Therefore, this shows that the relations obtained from~\eref{gen_rel_response_coef} for odd values of the index $m$ are all independent, which concludes the proof of the corollary.

We now investigate the consequences of this result on the total number of independent quantities $Q_{\alpha_1 \cdots \alpha_m \, , \, \beta_1 \cdots \beta_n}$ for a given total number $\mathcal{N}$ of subscripts.


\section{Counting of the independent coefficients \texorpdfstring{$Q_{\alpha_1 \cdots \alpha_m \, , \, \beta_1 \cdots \beta_n}$}{Q}}\label{number_ind_coef_sec}

We conclude this paper by counting the number of independent quantities $Q_{\alpha_1 \cdots \alpha_m \, , \, \beta_1 \cdots \beta_n}$ that are required to fully characterize the statistics of the currents in the nonequilibrium steady state. In particular, we use the results obtained in section~\ref{indep_rel_sec} to quantitatively investigate the consequence of the fluctuation relation (FR)~\eref{fluct_rel} on the properties of the full current statistics.

We want to determine the total number $N_{\chi}^{(\mathrm{ind})}(\mathcal{N})$ of independent quantities $Q_{\alpha_1 \cdots \alpha_m \, , \, \beta_1 \cdots \beta_n}$, for a fixed total number $\mathcal{N}=m+n \geqslant 1$ of subscripts, as a consequence of the FR. That is, we count the number of different quantities $Q_{\alpha_1 \cdots \alpha_m \, , \, \beta_1 \cdots \beta_n}$ for all values of the indices $m$ and $n$ that are compatible with the constraint that their sum $\mathcal{N}=m+n$ is fixed. As was already mentioned at the end of section~\ref{coef_sec}, considering a fixed number $\mathcal{N}$ of subscripts is justified by the structure of the general relation~\eref{gen_rel_response_coef}.

As discussed in section~\ref{coef_sec}, the quantities $Q_{\alpha_1 \cdots \alpha_m \, , \, \beta_1 \cdots \beta_n}$ are by construction the coefficients of the expansion~\eref{Q_exp_count_par_and_aff} of the cumulant generating function as a power series of both the counting parameters $\bm{\lambda}$ and the affinities $\bm{A}$. We recall that, in view of Einstein's convention used in~\eref{Q_exp_count_par_and_aff}, the subscripts $\alpha_i$ and $\beta_i$ of the quantity $Q_{\alpha_1 \cdots \alpha_m \, , \, \beta_1 \cdots \beta_n}$ all take values between 1 and $\chi$ ($\chi$ being the total number of currents). The set $\left\{ Q_{\alpha_1 \cdots \alpha_m \, , \, \beta_1 \cdots \beta_n} \right\}$, for all the integers $m,n \geqslant 0$ and for all the values $1 \leqslant \alpha , \beta \leqslant \chi$ of the subscripts, completely specifies the full current statistics in the nonequilibrium steady state. However, it must be noted that the elements of the latter set are not all independent. This rises from (i) the invariance~\eref{Q_inv_perm_indices} and (ii) the general relations~\eref{gen_rel_response_coef} generated by the FR, as we discuss in subsections~\ref{invariance_subsec} and~\ref{csqce_FR_subsec}, respectively. Finally, we illustrate in subsection~\ref{asympt_subsec} the results obtained in subsections~\ref{invariance_subsec} and~\ref{csqce_FR_subsec} by considering the asymptotic behavior of the number $N_{\chi}^{(\mathrm{ind})}(\mathcal{N})$ for $\mathcal{N} \gg 1$.


\subsection{Consequence of the invariance \texorpdfstring{\eref{Q_inv_perm_indices}}{}}\label{invariance_subsec}

We first consider fixed indices $m$ and $n$. \textit{A priori}, we have a total number $\chi^{m+n}$ of different quantities $Q_{\alpha_1 \cdots \alpha_m \, , \, \beta_1 \cdots \beta_n}$. However, this number is drastically reduced by the invariance property~\eref{Q_inv_perm_indices} under any permutation of the subscripts on the left or on the right of the comma. Indeed, the independent quantities $Q_{\alpha_1 \cdots \alpha_m \, , \, \beta_1 \cdots \beta_n}$ can then only be those that have different numbers of subscripts $\alpha$ or $\beta$ that are equal to $1, \ldots , \chi$.

We now count the total number $N_{\chi}^{(m,n)}$ of independent quantities $Q_{\alpha_1 \cdots \alpha_m \, , \, \beta_1 \cdots \beta_n}$ for given $m,n \geqslant 1$. To this end, we note here that $\mathscr{S}_{\chi}' \equiv \{ 1, \ldots,\chi \}$ defines a set of $\chi$ elements. The problem of finding $N_{\chi}^{(m,n)}$ is hence equivalent to determining the total numbers of \textit{unordered} collections $\{ \alpha_1,\ldots,\alpha_m \}$ and $\{ \beta_1,\ldots,\beta_n \}$ of $m$ and $n$ \textit{not necessarily distinct} elements of $\mathscr{S}_{\chi}'$, respectively. The latter numbers are known to be given respectively by $\binom{m+\chi-1}{\chi-1}$ and $\binom{n+\chi-1}{\chi-1}$ (see e.g. theorem~4.2 in~\cite{Rys}). The number $N_{\chi}^{(m,n)}$ is thus given by the product
\begin{eqnarray}
N_{\chi}^{(m,n)} = \binom{m+\chi-1}{\chi-1} \binom{n+\chi-1}{\chi-1} .
\label{N_chi_m_n}
\end{eqnarray}
Note that the above form of $N_{\chi}^{(m,n)}$ is also valid for $m=0$ or $n=0$. Therefore, we see that the invariance~\eref{Q_inv_perm_indices} reduces the number of different quantities $Q_{\alpha_1 \cdots \alpha_m \, , \, \beta_1 \cdots \beta_n}$, for given $m,n \geqslant 0$, from $\chi^{m+n}$ to~\eref{N_chi_m_n}.

We now denote by $N_{\chi}(\mathcal{N})$ the total number of different quantities $Q_{\alpha_1 \cdots \alpha_m \, , \, \beta_1 \cdots \beta_n}$, for a fixed total number $\mathcal{N}=m+n$ of subscripts and all compatible values of $m$ and $n$, with the invariance~\eref{Q_inv_perm_indices} taken into account. We recall that $Q_{\, , \, \beta_1 \cdots \beta_n} = 0$ for any $n \geqslant 0$ as a consequence of the normalization condition~\eref{cumulant_GF_norm_cond}. This means that the quantities $Q$ with $m=0$ are perfectly determined and hence do not need to enter our counting of independent quantities. With this in mind, the number $N_{\chi}(\mathcal{N})$, for any $\mathcal{N} \geqslant 1$, is thus given by, in view of~\eref{N_chi_m_n},
\begin{eqnarray}
\fl N_{\chi}(\mathcal{N}) = \sum_{m=1}^{\mathcal{N}} \, \sum_{n=0 \atop n+m=\mathcal{N}}^{\mathcal{N}-1} N_{\chi}^{(m,n)} = \sum_{m=1}^{\mathcal{N}} \, \sum_{n=0 \atop n+m=\mathcal{N}}^{\mathcal{N}-1} \binom{m+\chi-1}{\chi-1} \binom{n+\chi-1}{\chi-1} .
\label{nb_diff_Q_without_FR}
\end{eqnarray}

We emphasize that the quantity $N_{\chi}(\mathcal{N})$ only embeds the invariance~\eref{Q_inv_perm_indices}, and hence does not contain any information about the FR~\eref{fluct_rel}. We now discuss how the latter influences the total number of independent $Q_{\alpha_1 \cdots \alpha_m \, , \, \beta_1 \cdots \beta_n}$.


\subsection{Consequence of the fluctuation relation}\label{csqce_FR_subsec}

In section~\ref{coef_sec}, we saw that the FR~\eref{fluct_rel} generates general relations of the form~\eref{gen_rel_response_coef} between the quantities $Q_{\alpha_1 \cdots \alpha_m \, , \, \beta_1 \cdots \beta_n}$, for any $m,n \geqslant 0$. Pursuing the above analysis that led to~\eref{nb_diff_Q_without_FR}, we consider the relations that are obtained from~\eref{gen_rel_response_coef} for a given arbitrary total number $\mathcal{N}=m+n \geqslant 1$ of subscripts and for all the values of the indices $m$ and $n$ compatible with the constraint that their sum $\mathcal{N}$ is fixed. At first sight, the FR generates $\mathcal{N} + 1$ different types of relations: one for $m=\mathcal{N}$ and $n=0$, one for $m=\mathcal{N}-1$ and $n=1$, \ldots , and one for $m=0$ and $n=\mathcal{N}$. However, these $\mathcal{N} + 1$ relations are not all independent, as we discussed in details in the previous section.

Indeed, we showed in section~\ref{indep_rel_sec} that any relation obtained from~\eref{gen_rel_response_coef} with an arbitrary even index $m=m_K=2K$ can be deduced from the set of all the relations~\eref{gen_rel_response_coef} for every odd index $m>m_K$ (see theorem~\ref{theorem1}). This concerns the relations corresponding to the values $m=0$, $m=2$, \ldots , $m = 2 \mathbb{E} \left( \mathcal{N} / 2 \right)$. Hence we see that, among the $\mathcal{N} + 1$ relations \textit{a priori} obtained from~\eref{gen_rel_response_coef} for a given $\mathcal{N}$, $\mathbb{E} \left( \mathcal{N} / 2 \right) + 1$ of them can be derived from the remaining ones. Therefore, our analysis shows that the number of independent relations between the quantities $Q_{\alpha_1 \cdots \alpha_m \, , \, \beta_1 \cdots \beta_n}$ that are generated by the FR~\eref{fluct_rel} is only $\mathbb{E} \left[ \left( \mathcal{N}+1 \right) / 2 \right]$. These independent relations are the ones obtained from~\eref{gen_rel_response_coef} for odd values of the index $m$, i.e. for $m=1$, $m=3$, \ldots , and $m = 2 \mathbb{E} \left[ \left( \mathcal{N}+1 \right) / 2 \right] - 1$.

A single one of these $\mathbb{E} \left[ \left( \mathcal{N}+1 \right) / 2 \right]$ independent relations corresponds to a particular odd value of $m$, and expresses the coefficient $Q_{\alpha_1 \cdots \alpha_m \, , \, \beta_1 \cdots \beta_n}$ in terms of quantities $Q_{\alpha_1 \cdots \alpha_{m'} \, , \, \beta_1 \cdots \beta_{n'}}$ that have $m'=m+1$, \ldots , $m+n$ subscripts on the left of the comma. Since the subscripts $\alpha_i$ and $\beta_j$ take values between 1 and $\chi$, there is one such relation for all possible values of $\alpha_i$ and $\beta_j$. We already saw in section~\ref{invariance_subsec} that there are $N_{\chi}^{(m,n)}$ different $Q_{\alpha_1 \cdots \alpha_m \, , \, \beta_1 \cdots \beta_n}$ for fixed indices $m$ and $n$, as a consequence of the invariance~\eref{Q_inv_perm_indices}. Therefore, one relation obtained from~\eref{gen_rel_response_coef} for a given odd index $m$ actually yields $N_{\chi}^{(m,n)}$ different relations. Denoting by $N_{\chi}^{(\mathrm{FR})}(\mathcal{N})$ the total number, for a fixed number $\mathcal{N} \geqslant 1$ of subscripts, of independent relations between quantities $Q_{\alpha_1 \cdots \alpha_m \, , \, \beta_1 \cdots \beta_n}$ generated by the FR~\eref{fluct_rel}, we hence have, in view of the expression~\eref{N_chi_m_n},
\begin{eqnarray}
\fl N_{\chi}^{(\mathrm{FR})}(\mathcal{N}) = \sum_{m=1 \atop m \, \mathrm{odd}}^{\mathcal{N}} \; \sum_{n=0 \atop n+m=\mathcal{N}}^{\mathcal{N}-1} N_{\chi}^{(m,n)} = \sum_{m=1 \atop m \, \mathrm{odd}}^{\mathcal{N}} \; \sum_{n=0 \atop n+m=\mathcal{N}}^{\mathcal{N}-1} \binom{m+\chi-1}{\chi-1} \binom{n+\chi-1}{\chi-1} .
\label{nb_diff_rel_from_FR}
\end{eqnarray}

Therefore, we first determined in subsection~\ref{invariance_subsec} the number $N_{\chi}(\mathcal{N})$ of different quantities $Q_{\alpha_1 \cdots \alpha_m \, , \, \beta_1 \cdots \beta_n}$ for a fixed total number $\mathcal{N}$ of subscripts. This number was deduced from the invariance~\eref{Q_inv_perm_indices}, and thus represents the total number of quantities $Q_{\alpha_1 \cdots \alpha_m \, , \, \beta_1 \cdots \beta_n}$ that need to be known independently of the FR~\eref{fluct_rel}. Then we found the number $N_{\chi}^{(\mathrm{FR})}(\mathcal{N})$, for a fixed $\mathcal{N}$ as well, of independent relations generated by the FR. We are now in position to evaluate the total number $N_{\chi}^{(\mathrm{ind})}(\mathcal{N})$ of independent quantities $Q_{\alpha_1 \cdots \alpha_m \, , \, \beta_1 \cdots \beta_n}$, for a given $\mathcal{N} \geqslant 1$, as a consequence of the FR. It must be merely given by the difference
\begin{eqnarray*}
N_{\chi}^{(\mathrm{ind})}(\mathcal{N}) = N_{\chi}(\mathcal{N}) - N_{\chi}^{(\mathrm{FR})}(\mathcal{N})
\end{eqnarray*}
that is, in view of the expressions~\eref{nb_diff_Q_without_FR} and~\eref{nb_diff_rel_from_FR} of $N_{\chi}(\mathcal{N})$ and $N_{\chi}^{(\mathrm{FR})}(\mathcal{N})$, respectively,
\begin{eqnarray}
N_{\chi}^{(\mathrm{ind})}(\mathcal{N}) = \sum_{m=1 \atop m \, \mathrm{even}}^{\mathcal{N}} \; \sum_{n=0 \atop n+m=\mathcal{N}}^{\mathcal{N}-1} \binom{m+\chi-1}{\chi-1} \binom{n+\chi-1}{\chi-1} .
\label{nb_diff_Q_with_FR}
\end{eqnarray}
This quantity describes the actual number of independent quantities $Q_{\alpha_1 \cdots \alpha_m \, , \, \beta_1 \cdots \beta_n}$ that, in view of the underlying FR~\eref{fluct_rel}, yield the full information regarding the underlying statistics of the currents in the nonequilibrium steady state. These independent quantities correspond to those that remain undetermined by the FR, and that need to be specified (either theoretically or experimentally) in any particular physical problem.

The expressions~\eref{N_chi_m_n}-\eref{nb_diff_Q_with_FR} are valid for any $\mathcal{N} \geqslant 1$ and any $\chi \geqslant 1$. We now illustrate these results by considering their asymptotic behavior for $\mathcal{N} \gg 1$.


\subsection{Asymptotic behavior for \texorpdfstring{$\mathcal{N} \gg 1$}{N large}}\label{asympt_subsec}

Here we apply the results of subsections~\ref{invariance_subsec} and~\ref{csqce_FR_subsec} regarding the total number of independent quantities $Q_{\alpha_1 \cdots \alpha_m \, , \, \beta_1 \cdots \beta_n}$ to the case of a large total number $\mathcal{N}=m+n$ of subscripts, i.e. $\mathcal{N} \gg 1$. We consider an arbitrary total number $\chi$ of currents, $\chi \geqslant 1$.

We first investigate the smallest values of $\chi$, as depicted in figure~\ref{fig1}. Figure~\ref{fig1}(a) shows the behavior of the total number $N_{\chi}$ of different quantities $Q_{\alpha_1 \cdots \alpha_m \, , \, \beta_1 \cdots \beta_n}$ with the invariance~\eref{Q_inv_perm_indices} and the total number $N_{\chi}^{(\mathrm{ind})}$ of independent $Q_{\alpha_1 \cdots \alpha_m \, , \, \beta_1 \cdots \beta_n}$ with the FR~\eref{fluct_rel} taken into account, both as functions of $\mathcal{N}$ for different values of $\chi$. Figure~\ref{fig1}(b) then shows the behavior of the ratio $N_{\chi}^{(\mathrm{ind})}/N_{\chi}$ as a function of $\mathcal{N}$ for the same values of $\chi$. We readily see that, independently of the value of $\chi$, this ratio tends to the value 1/2 for large values of $\mathcal{N}$.

\begin{figure*}[ht]
\centering
\includegraphics[width=5.5in]{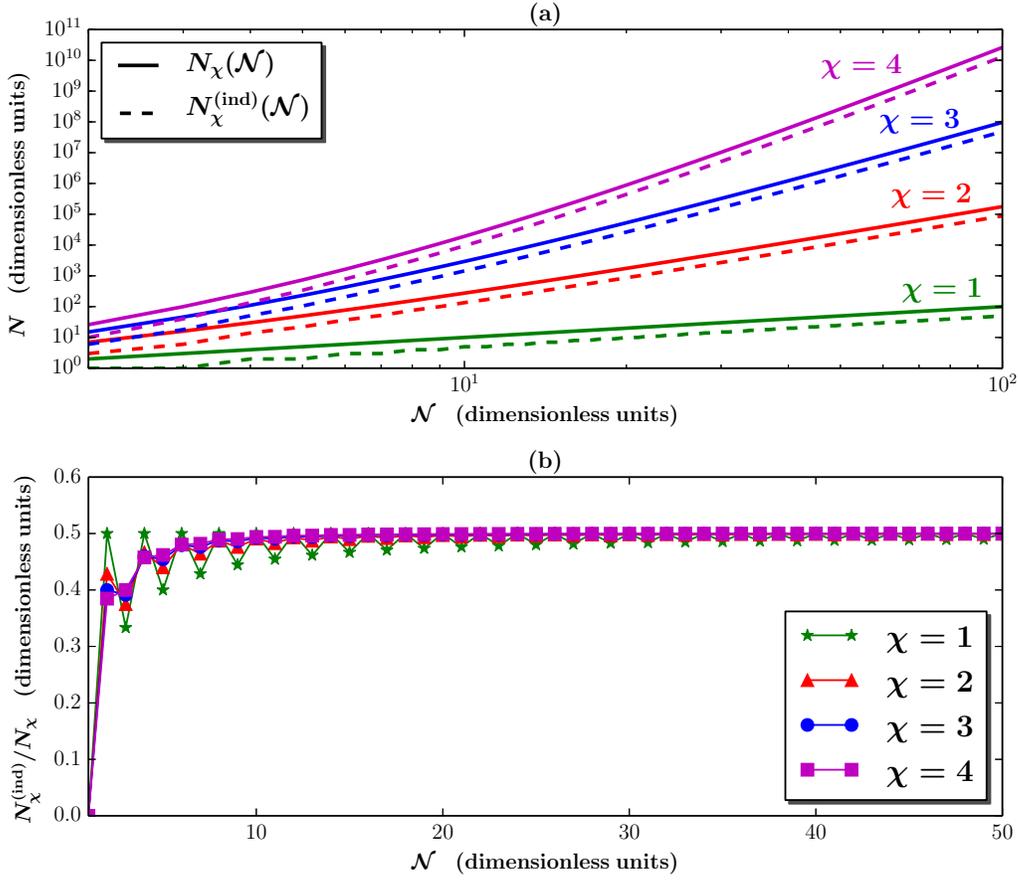}
\caption{(Color online) (a) Behavior of the numbers \eref{nb_diff_Q_without_FR} (solid curves) and \eref{nb_diff_Q_with_FR} (dotted curves) as functions of $\mathcal{N}$ for different values of $\chi$. (b) Behavior of the ratio $N_{\chi}^{(\mathrm{ind})}/N_{\chi}$ as a function of $\mathcal{N}$ for the same values of $\chi$.}
\label{fig1}
\end{figure*}

\begin{table}[ht]
\centering
\caption{Asymptotic behavior of the numbers \eref{nb_diff_Q_without_FR} and \eref{nb_diff_Q_with_FR}  for $\mathcal{N}\gg1$ for the same values of $\chi$ as in figure~\ref{fig1}.}
\vskip 0.2 cm
\begin{tabular}{ |M{3cm}|M{3cm}|M{3cm}|N  }
 \hline
 $\chi$ & $N_{\chi} (\mathcal{N} \gg 1)$ & $N_{\chi}^{(\mathrm{ind})} (\mathcal{N} \gg 1)$ & \\[0.5cm]
 \hline\hline
 1   & $\mathcal{N}$  & $\mathcal{N}/2$ & \\[0.2cm]
 2   & $\mathcal{N}^3/6$  & $\mathcal{N}^3/12$ & \\[0.2cm]
 3   & $\mathcal{N}^5/30$  & $\mathcal{N}^5/60$ & \\[0.2cm]
 4   & $\mathcal{N}^7/5040$  & $\mathcal{N}^7/10080$ & \\[0.2cm]
$\vdots$   & $\vdots$  & $\vdots$ & \\[0.2cm]
 \hline
\end{tabular}
\label{table1}
\end{table}

This behavior of the ratio $N_{\chi}^{(\mathrm{ind})}/N_{\chi}$ for large values of $\mathcal{N}$ can be calculated using the expressions~\eref{nb_diff_Q_without_FR} and~\eref{nb_diff_Q_with_FR} of $N_{\chi}$ and $N_{\chi}^{(\mathrm{ind})}$, respectively, giving polynomial functions of $\mathcal{N}$. Keeping the highest power of $\mathcal{N}$, the asymptotic behavior of $N_{\chi}$ and $N_{\chi}^{(\mathrm{ind})}$ is obtained in the limit $\mathcal{N} \gg 1$, as summarized in table~\ref{table1} for the same values of $\chi$ as in figure~\ref{fig1}. This confirms the observation drawn from figure~\ref{fig1}(b) that the ratio $N_{\chi}^{(\mathrm{ind})}/N_{\chi}$ indeed tends to 1/2 in the limit $\mathcal{N} \gg 1$ for $\chi = 1, \ldots , 4$.

The foregoing analysis suggests that the ratio $N_{\chi}^{(\mathrm{ind})}/N_{\chi}$ tends, in the limit $\mathcal{N} \gg 1$, to the value 1/2 for an arbitrary value of $\chi$. This is precisely the content of the following theorem, which we prove in the remaining part of this subsection:

\begin{theorem}
The ratio of the quantities $N_{\chi}^{(\mathrm{ind})}$ and $N_{\chi}$, respectively given by~\eref{nb_diff_Q_with_FR} and~\eref{nb_diff_Q_without_FR}, satisfies the asymptotic behavior
\begin{eqnarray}
\frac{N_{\chi}^{(\mathrm{ind})}(\mathcal{N})}{N_{\chi}(\mathcal{N})} \approx \frac{1}{2} \qquad\mbox{for}\qquad \mathcal{N} \gg 1 \, 
\label{ratio_asympt}
\end{eqnarray}
and an arbitrary value of $\chi$.
\label{theorem2}
\end{theorem}

The main idea that underlies the proof of this theorem is to note that the defining sums~\eref{nb_diff_Q_without_FR} and~\eref{nb_diff_Q_with_FR} for the numbers $N_{\chi}$ and $N_{\chi}^{(\mathrm{ind})}$, respectively, can be approximated by integrals in the limit $\mathcal{N} \gg 1$. To this end, it proves useful to rewrite~\eref{nb_diff_Q_without_FR} and~\eref{nb_diff_Q_with_FR} in the form
\begin{eqnarray}
N_{\chi}(\mathcal{N}) = \frac{1}{\left[ (\chi-1)! \right]^2} \, \sum_{m=1}^{\mathcal{N}} f_{\chi} (m)
\label{N_chi_expr}
\end{eqnarray}
and
\begin{eqnarray}
N_{\chi}^{(\mathrm{ind})}(\mathcal{N}) = \frac{1}{\left[ (\chi-1)! \right]^2} \sum_{m=1}^{\mathbb{E} \left( \frac{\mathcal{N}}{2} \right)} f_{\chi} (2m) 
\label{N_chi_ind_expr}
\end{eqnarray}
where $f_{\chi}$ is the polynomial function defined by
\begin{eqnarray}
f_{\chi}(x) \equiv \prod_{k=1}^{\chi-1} (x+k) (\mathcal{N}-x+k) 
\label{f_chi_def}
\end{eqnarray}
for any $\chi \geqslant 2$.  The fact that this function is here defined for $\chi \geqslant 2$ is not a problem regarding the subsequent proof, since theorem~\ref{theorem2} has already been proved for $\chi=1$ in table~\ref{table1}. From now on, we consider an arbitrary but fixed number $\chi$ of currents.

It is worth pointing out that we can rewrite~\eref{f_chi_def} as 
\begin{eqnarray}
f_{\chi}(x) = \prod_{k=1}^{\chi-1} \left[ \left( \frac{\mathcal{N}}{2} + k \right)^2 - \left(\frac{\mathcal{N}}{2}-x \right)^2 \right] .
\label{f_chi_sym_form}
\end{eqnarray}
This readily shows that the function $f_{\chi}(x)$ is symmetric under the transformation $x\to\mathcal{N}-x$ and is maximal at the point $x=\mathcal{N}/2$. Therefore, the structure of this function ensures that, for $\mathcal{N} \gg 1$, the main contributions to the two sums involved in~\eref{N_chi_expr} and~\eref{N_chi_ind_expr} arise from values of $m$ and $2m$, respectively, in the neighborhood of $\mathcal{N}/2$ where
\begin{eqnarray}
f_{\chi} \left( \mathcal{N}/2 \right) = \left[ \prod_{k=1}^{\chi-1} \left( \frac{\mathcal{N}}{2}+k \right) \right]^2 \sim \mathcal{N}^{2\chi-2}
\label{f_chi_N_ov_2_expr}
\end{eqnarray}
for a fixed value of $\chi$ and $\mathcal{N} \gg 1$. In contrast, the values near the boundaries $1$ and $\mathcal{N}$ do not significantly contribute to the sums because
\begin{eqnarray}
f_{\chi}(1) \sim \mathcal{N}^{\chi-1} \qquad\mbox{and}\qquad
f_{\chi}(\mathcal{N}) \sim \mathcal{N}^{\chi-1} \, .
\label{f_chi_1+N_expr}
\end{eqnarray}
It is then clear on~\eref{f_chi_N_ov_2_expr}-\eref{f_chi_1+N_expr} that $f_{\chi}(\mathcal{N}/2) \gg f_{\chi}(1)$ and $f_{\chi}(\mathcal{N}/2) \gg f_{\chi}(\mathcal{N})$ for a fixed $\chi$ and $\mathcal{N} \gg 1$.

Now, we discuss how the discrete sums in~\eref{N_chi_expr}-\eref{N_chi_ind_expr} can be approximated by integrals for $\mathcal{N} \gg 1$.  We begin with the following:

\begin{lemma}
The numbers $N_{\chi}$ and $N_{\chi}^{(\mathrm{ind})}$, respectively given by~\eref{N_chi_expr} and~\eref{N_chi_ind_expr}, can be approximated by the integrals
\begin{eqnarray}
N_{\chi}(\mathcal{N}) \approx \frac{1}{\left[ (\chi-1)! \right]^2} \int_{1}^{\mathcal{N}} f_{\chi}(x) \, dx
\label{N_chi_int_expr}
\end{eqnarray}
and
\begin{eqnarray}
N_{\chi}^{(\mathrm{ind})}(\mathcal{N}) \approx \frac{1}{\left[ (\chi-1)! \right]^2} \int_{1}^{\mathbb{E} \left( \frac{\mathcal{N}}{2} \right)} f_{\chi}(2x) \, dx
\label{N_chi_ind_int_expr}
\end{eqnarray}
for an arbitrary value of $\chi$ and $\mathcal{N} \gg 1$.
\label{lemma_asympt1}
\end{lemma}

The proof below makes use of the well-known Euler summation formula \cite{Apo99} (also referred to as the Euler-Maclaurin formula \cite{AbrSteg}),
\begin{eqnarray}
\fl \sum_{j=1}^{k} f(j) = \int_{1}^{k} f(x) \, dx + \frac{1}{2} [f(1)+f(k)] + \sum_{j=1}^{p} \frac{B_{2j}}{(2j)!} \left[ f^{(2j-1)}(k) - f^{(2j-1)}(1) \right] \nonumber\\[0.1cm]
+ \frac{1}{(2p+1)!} \int_{1}^{k} B_{2p+1}\left( x-\mathbb{E}(x) \right) f^{(2p+1)}(x) \, dx
\label{Euler_formula}
\end{eqnarray}
where $k$ and $p$ are two integers such that $k,p \geqslant 1$, $f(x)$ is an arbitrary function that admits $2p+1$ continuous derivatives and $f^{(j)}$ is the $j^{\mathrm{th}}$ derivative of $f$. Here $B_j$ and $B_j(x)$ denote the $j^{\mathrm{th}}$ Bernoulli number and polynomial, respectively, while $B_{j}\left( x-\mathbb{E}(x) \right)$ is referred to as the $j^{\mathrm{th}}$ Bernoulli periodic function \cite{AbrSteg}. We emphasize that the formula~\eref{Euler_formula} is \textit{exact}, and that the three last terms in the right-hand side provide the exact error made when estimating the discrete sum by the integral.

\begin{proof}
For clarity, we begin with the proof of the approximation~\eref{N_chi_int_expr} of $N_{\chi}$. We apply the Euler formula~\eref{Euler_formula} to the function $f_{\chi}$, which is a polynomial function of degree $2\chi-2$, as readily seen on its definition~\eref{f_chi_def}. We can thus choose $p=\chi-1$ into the Euler formula so that the last integral in the right-hand side of~\eref{Euler_formula} is equal to zero since $f_{\chi}^{(2\chi-1)}$ vanishes identically. Setting $k=\mathcal{N}$ into~\eref{Euler_formula} hence yields
\begin{eqnarray}
\fl \sum_{j=1}^{\mathcal{N}} f_{\chi}(j) = \int_{1}^{\mathcal{N}} f_{\chi}(x) \, dx + \frac{1}{2}[f_{\chi}(1)+f_{\chi}(\mathcal{N})] + \sum_{j=1}^{\chi-1} \frac{B_{2j}}{(2j)!} \left[ f_{\chi}^{(2j-1)}(\mathcal{N}) - f_{\chi}^{(2j-1)}(1) \right] . \nonumber\\
\label{Euler_formula_f_chi}
\end{eqnarray}
Now, we show that the second and third terms in the right-hand side of~\eref{Euler_formula_f_chi} can be neglected, in the limit $\mathcal{N} \gg 1$, compared to the integral. To this end, it proves useful to find a relevant lower bound for the integral.

Since $x+k \geqslant x$ and $\mathcal{N}-x+k \geqslant 1$ hold for any $1 \leqslant x \leqslant \mathcal{N}$ and $1 \leqslant k \leqslant \chi-1$, we see that the function~\eref{f_chi_def} satisfies $f_{\chi}(x) \geqslant x^{\chi-1}$ for any $1 \leqslant x \leqslant \mathcal{N}$.  Consequently, we get the inequality
\begin{eqnarray}
\int_{1}^{\mathcal{N}} f_{\chi}(x) \, dx \geqslant \int_{1}^{\mathcal{N}} x^{\chi-1} \, dx \sim \mathcal{N}^{\chi}
\label{int_1_to_N_f_chi_ineq}
\end{eqnarray}
for $\mathcal{N} \gg 1$, so that the integral $\int_{1}^{\mathcal{N}} f_{\chi}(x) \, dx$ behaves \textit{at least} as $\mathcal{N}^{\chi}$. Since we already know from~\eref{f_chi_1+N_expr} that $f_{\chi}(1),f_{\chi}(\mathcal{N}) \sim \mathcal{N}^{\chi-1}$, we readily see that
\begin{eqnarray}
f_{\chi}(1), f_{\chi}(\mathcal{N}) \ll \int_{1}^{\mathcal{N}} f_{\chi}(x) \, dx 
\label{f_chi_1_N_neglig}
\end{eqnarray}
for $\mathcal{N} \gg 1$. This shows that the term $[f_{\chi}(1)+f_{\chi}(\mathcal{N})]/2$ in the right-hand side of~\eref{Euler_formula_f_chi} can be neglected for $\mathcal{N} \gg 1$, compared to the integral.

Now, we use the exact same line of reasoning to neglect the third term in the right-hand side of~\eref{Euler_formula_f_chi}. Similarly to~\eref{f_chi_1+N_expr}, it can be readily seen that the derivatives $f_{\chi}^{(2j-1)}(1)$ and $f_{\chi}^{(2j-1)}(\mathcal{N})$ behave \textit{at most} as $\mathcal{N}^{\chi-1}$ for any $1 \leqslant j \leqslant \chi-1$. It is then clear from~\eref{int_1_to_N_f_chi_ineq} that for $\mathcal{N} \gg 1$ we have
\begin{eqnarray}
f_{\chi}^{(2j-1)}(1), f_{\chi}^{(2j-1)}(\mathcal{N}) \ll \int_{1}^{\mathcal{N}} f_{\chi}(x) \, dx
\label{f_chi_deriv_neglig}
\end{eqnarray}
for any $1 \leqslant j \leqslant \chi-1$. This shows that the third term in the right-hand side of~\eref{Euler_formula_f_chi} can also be neglected, for $\mathcal{N} \gg 1$, compared to the integral so that~\eref{Euler_formula_f_chi} merely reads
\begin{eqnarray}
\sum_{j=1}^{\mathcal{N}} f_{\chi}(j) \approx \int_{1}^{\mathcal{N}} f_{\chi}(x) \, dx \, .
\label{Euler_formula_f_chi_final}
\end{eqnarray}
Combining~\eref{Euler_formula_f_chi_final} with the expression~\eref{N_chi_expr} of $N_{\chi}$ hence yields the desired result~\eref{N_chi_int_expr}.

The proof of the approximation~\eref{N_chi_ind_int_expr} of $N_{\chi}^{(\mathrm{ind})}$ proceeds in the same way by distinguishing the two cases of $\mathcal{N}$ even and $\mathcal{N}$ odd for the integer part $\mathbb{E}(\mathcal{N}/2)$.
\end{proof}

Now, the approximation~\eref{N_chi_ind_int_expr} of $N_{\chi}^{(\mathrm{ind})}$ can be rewritten without the integer part $\mathbb{E}(\mathcal{N}/2)$ using the following:

\begin{lemma}
The number $N_{\chi}^{(\mathrm{ind})}$ given by~\eref{N_chi_ind_expr} can be approximated by the integral
\begin{eqnarray}
N_{\chi}^{(\mathrm{ind})}(\mathcal{N}) \approx \frac{1}{2} \, \frac{1}{\left[ (\chi-1)! \right]^2} \int_{1}^{\mathcal{N}} f_{\chi}(x) \, dx
\label{N_chi_ind_int_final_expr}
\end{eqnarray}
for an arbitrary value of $\chi$ and $\mathcal{N} \gg 1$.
\label{lemma_asympt2}
\end{lemma}

\begin{proof}
The integer part $\mathbb{E}(\mathcal{N}/2)$ in the integral of~\eref{N_chi_ind_int_expr} requires us to distinguish the two cases of $\mathcal{N}$ even and $\mathcal{N}$ odd. Supposing $\mathcal{N}$ even, we have $\mathbb{E}(\mathcal{N}/2)=\mathcal{N}/2$, so that the integral of~\eref{N_chi_ind_int_expr} reads
\begin{eqnarray}
\int_{1}^{\mathbb{E} \left( \frac{\mathcal{N}}{2} \right)} f_{\chi}(2x) \, dx = \frac{1}{2} \left[ \int_{1}^{\mathcal{N}} f_{\chi}(x) \, dx - \int_{1}^{2} f_{\chi}(x) \, dx \right]
\label{int_rel1}
\end{eqnarray}
after the change of variable $x \to 2x$. The second integral in the right-hand side of~\eref{int_rel1} can be neglected compared to the first one. Indeed, we note that the function~\eref{f_chi_def} satisfies $f_{\chi}(x) \leqslant (\chi+1)^{\chi-1}(\mathcal{N}+\chi-2)^{\chi-1}$ for $1 \leqslant x \leqslant 2$, from which we immediately get the inequality
\begin{eqnarray}
\int_{1}^{2} f_{\chi}(x) \, dx \leqslant (\chi+1)^{\chi-1}(\mathcal{N}+\chi-2)^{\chi-1} \sim \mathcal{N}^{\chi-1}
\label{int_1_to_2_f_chi_ineq}
\end{eqnarray}
for $\mathcal{N} \gg 1$. This upper bound shows that the integral $\int_{1}^{2} f_{\chi}(x) \, dx$ behaves \textit{at most} as $\mathcal{N}^{\chi-1}$. Since we already know from~\eref{int_1_to_N_f_chi_ineq} that $\int_{1}^{\mathcal{N}} f_{\chi}(x) \, dx$ behaves \textit{at least} as $\mathcal{N}^{\chi}$, we readily obtain
\begin{eqnarray}
\int_{1}^{2} f_{\chi}(x) \, dx \ll \int_{1}^{\mathcal{N}} f_{\chi}(x) \, dx
\label{int_1_to_2_neglig}
\end{eqnarray}
for $\mathcal{N} \gg 1$. Finally, combining~\eref{int_rel1} and~\eref{int_1_to_2_neglig} with the approximation~\eref{N_chi_ind_int_expr} of $N_{\chi}^{(\mathrm{ind})}$ yields the desired result~\eref{N_chi_ind_int_final_expr}. Lemma~\ref{lemma_asympt2} is thus proved for $\mathcal{N}$ even.

The proof proceeds in the same way for the case of an odd $\mathcal{N}$. The only difference is that the previous expression~\eref{int_rel1} must here be replaced by
\begin{eqnarray}
\fl \int_{1}^{\mathbb{E} \left( \frac{\mathcal{N}}{2} \right)} f_{\chi}(2x) \, dx = \frac{1}{2} \left[ \int_{1}^{\mathcal{N}} f_{\chi}(x) \, dx - \int_{1}^{2} f_{\chi}(x) \, dx - \int_{\mathcal{N}-1}^{\mathcal{N}} f_{\chi}(x) \, dx \right] .
\label{int_rel2}
\end{eqnarray}
Since the function~\eref{f_chi_def} satisfies $f_{\chi}(x) \leqslant \chi^{\chi-1}(\mathcal{N}+\chi-1)^{\chi-1}$ for $\mathcal{N}-1 \leqslant x \leqslant \mathcal{N}$, we also have that
\begin{eqnarray}
\int_{\mathcal{N}-1}^{\mathcal{N}} f_{\chi}(x) \, dx \ll \int_{1}^{\mathcal{N}} f_{\chi}(x) \, dx 
\label{int_N_min_1_to_N_neglig}
\end{eqnarray}
for $\mathcal{N} \gg 1$. Combining~\eref{int_1_to_2_neglig}-\eref{int_N_min_1_to_N_neglig} with the approximation~\eref{N_chi_ind_int_expr} of $N_{\chi}^{(\mathrm{ind})}$ hence readily yields the desired result~\eref{N_chi_ind_int_final_expr}, which also proves lemma~\ref{lemma_asympt2} for $\mathcal{N}$ odd.
\end{proof}

With the results of lemmas~\ref{lemma_asympt1} and~\ref{lemma_asympt2} at hand, it is now straightforward to derive the asymptotic behavior of the ratio $N_{\chi}^{(\mathrm{ind})}/N_{\chi}$. Indeed, the approximations~\eref{N_chi_int_expr} and~\eref{N_chi_ind_int_final_expr} for $N_{\chi}$ and $N_{\chi}^{(\mathrm{ind})}$, respectively, lead precisely to the result~\eref{ratio_asympt}. This concludes the proof of theorem~\ref{theorem2}.

Therefore, the main outcome of theorem~\ref{theorem2} is to show that the total number of independent quantities $Q_{\alpha_1 \cdots \alpha_m \, , \, \beta_1 \cdots \beta_n}$ that possess $\mathcal{N}=m+n$ subscripts is, for large values of $\mathcal{N}$, effectively divided by two as a consequence of the FR~\eref{fluct_rel}. This result is valid for an arbitrary total number $\chi \geqslant 1$ of currents. As seen in figure~\ref{fig1}, this result already gives an accurate estimate for values of $\mathcal{N}$ of the order of magnitude of ten to a hundred. The larger the value of $\chi$, the smaller the value of $\mathcal{N}$ for the approximation~\eref{ratio_asympt} to be accurate.


\section{Conclusion}\label{conclusion_sec}

In this paper we studied the influence of the multivariate fluctuation relation (FR) expressed by~\eref{fluct_rel}, and thus of microreversibility, on the statistics of the currents flowing across general open systems in nonequilibrium steady states. We quantitatively showed that the FR greatly reduces the total number of independent quantities that need to be specified in order to fully determine this statistics.

The latter is described by the generating function $Q \left( \bm{\lambda} , \bm{A} \right)$ of the statistical cumulants of the currents. The $m^{\mathrm{th}}$ cumulant $Q_{\alpha_1 \cdots \alpha_m} \left( \bm{A} \right)$ is defined by the derivative of the generating function with respect to the counting parameters $\lambda_{\alpha_1}$, \ldots , $\lambda_{\alpha_m}$. The $n^{\mathrm{th}}$ response $Q_{\alpha_1 \cdots \alpha_m \, , \, \beta_1 \cdots \beta_n}$ of the $m^{\mathrm{th}}$ cumulant is then defined by the derivative of $Q_{\alpha_1 \cdots \alpha_m} \left( \bm{A} \right)$ with respect to the affinities $A_{\beta_1}$, \ldots , $A_{\beta_n}$. The statistics of the currents can thus be alternatively described by the set $\left\{ Q_{\alpha_1 \cdots \alpha_m \, , \, \beta_1 \cdots \beta_n} \right\}$.

We saw that expanding both sides of the FR~\eref{fluct_rel} as power series of both $\bm{\lambda}$ and $\bm{A}$ generates relations of the form~\eref{gen_rel_response_coef} between the quantities $Q_{\alpha_1 \cdots \alpha_m \, , \, \beta_1 \cdots \beta_n}$. We were then able to identify the independent elements of the complete set $\left\{ Q_{\alpha_1 \cdots \alpha_m \, , \, \beta_1 \cdots \beta_n} \right\}$ by investigating all relations obtained from~\eref{gen_rel_response_coef} for a fixed total number $\mathcal{N} = m+n$ of subscripts and all compatible values of the indices $m$ and $n$.

We thus showed that, for a fixed $\mathcal{N} = m+n$, the relation~\eref{gen_rel_response_coef} that corresponds to an arbitrary even index $m=m_K \equiv 2K$ can be deduced from the set of all relations obtained from~\eref{gen_rel_response_coef} for all odd indices $m>m_K$. Therefore, the main outcome of our analysis is that the only independent relations generated by the FR are the relations~\eref{gen_rel_response_coef} that correspond to odd indices $m$. We then used this result to count the number of independent quantities $Q_{\alpha_1 \cdots \alpha_m \, , \, \beta_1 \cdots \beta_n}$ that fully characterize the current statistics. In the limit case of a large total number $\mathcal{N}$ of subscripts, $\mathcal{N} \gg 1$, we saw that the effect of the FR is to effectively divide by two the total number of these independent quantities.

We believe our work shed new light on the general properties of the statistics of currents in nonequilibrium physical systems that obey a FR. We identified the quantities that are left unspecified by the general mathematical structure that governs the system, namely the FR. These quantities remain to be determined in any particular physical problem so as to fully characterize the statistics of the currents. Our results may thus e.g. be of interest for the experimental study of nonequilibrium systems that obey a FR. Indeed, they clearly indicate which quantities must be measured experimentally from the full current statistics. As another example, our study may also prove useful in the field of electron full counting statistics (see e.g. \cite{Naz}).

We recall that the present paper deals with nonequilibrium systems that operate in the absence of magnetic field. A potentially interesting direction of further research could for instance be to investigate how the above results generalize to systems with a nonzero magnetic field. Such an extension would be particularly relevant regarding e.g. quantum transport in nanostructures or mesoscopic conductors.


\section*{Acknowledgments}

This research is financially supported by the Universit\'e libre de Bruxelles (ULB) and the Fonds de la Recherche Scientifique~-~FNRS under the Grant PDR~T.0094.16 for the project ``SYMSTATPHYS".


\appendix


\section{Proof of the relation \texorpdfstring{\eref{gen_rel_response_coef}}{}}\label{proof_rel_app}

As was outlined in the previous section, the first step is to expand both sides of the FR~\eref{fluct_rel} in the alternative form~\eref{fluct_rel_2} as power series of both the counting parameters $\bm{\lambda}$ and the affinities $\bm{A}$. We then merely identify  the coefficients of a same power of both the counting parameters and the affinities, and readily obtain the general relation~\eref{gen_rel_response_coef}.

In view of~\eref{Q_exp_count_par_and_aff}, the left-hand side of~(\ref{fluct_rel_2}) is expanded as
\begin{eqnarray}
\fl Q \left( -\bm{\lambda} , \bm{A} \right) = \sum_{m = 0}^{\infty}\sum_{n = 0}^{\infty} \frac{(-1)^m}{m! n!} \, Q_{\alpha_1 \cdots \alpha_m \, , \, \beta_1 \cdots \beta_n} \lambda_{\alpha_1} \cdots \lambda_{\alpha_m} A_{\beta_1} \cdots A_{\beta_n} \, .
\label{Q_minLbda_exp_count_par_and_aff}
\end{eqnarray}
We now must write the right-hand side of~\eref{fluct_rel_2} as a power series of $\bm{\lambda}$ and $\bm{A}$.
Expanding the translation operator~(\ref{transl}) acting on the cumulant generating function yields
\begin{eqnarray}
\hat T \left( \bm{A} \right) Q \left( \bm{\lambda} , \bm{A} \right) = \sum_{k = 0}^{\infty} \frac{1}{k!} \left( \bm{A} \cdot \frac{\partial}{\partial \bm{\lambda}} \right)^k Q \left( \bm{\lambda} , \bm{A} \right) .
\label{action_T_on_Q_series_for_T}
\end{eqnarray}

According to Newton's chain rule for the derivatives of products, we have that
\begin{eqnarray}
\left( \bm{A} \cdot \frac{\partial}{\partial \bm{\lambda}} \right) \, \lambda_{\alpha_1}\cdots \lambda_{\alpha_m} = \sum_{j=1}^m \lambda_{\alpha_1}\cdots \lambda_{\alpha_{j-1}} A_{\alpha_j} \lambda_{\alpha_{j+1}}\cdots \lambda_{\alpha_m} 
\end{eqnarray}
so that the degree in the counting parameters $\bm{\lambda}$ is decreased by one, while the degree in the affinities $\bm{A}$ is increased by one.  If we relabel the index of $A_{\alpha_j}$ into $A_{\beta_j}$ in order to recover the chosen notation, we see that the index ${\alpha_j}$ has been replaced by ${\beta_j}$.  Therefore, the translation operator has the action of moving the indices $\beta_j$ from the right to the left of the comma in the adopted notations.  Consequently, the action of the operator $\left( \bm{A} \cdot \partial / \partial \bm{\lambda} \right)^k$ on the power series~\eref{Q_exp_count_par_and_aff} of the cumulant GF $Q$ is given by
\begin{eqnarray}
\fl \left( \bm{A} \cdot \frac{\partial}{\partial \bm{\lambda}} \right)^k Q \left( \bm{\lambda} , \bm{A} \right) = \sum_{m=0}^{\infty}\sum_{n=k}^{\infty} \frac{1}{m! n!} \, Q_{\alpha_1 \cdots \alpha_m \, , \, \beta_1 \cdots \beta_n}^{\{k\}} \lambda_{\alpha_1} \cdots \lambda_{\alpha_m} A_{\beta_1} \cdots A_{\beta_n} 
\label{A_partLbda_k_on_Q}
\end{eqnarray}
for any $k \geqslant 0$, and where the quantities $Q^{\{k\}}$ are defined by
\begin{eqnarray}
\fl Q_{\alpha_1 \cdots \alpha_m \, , \, \beta_1 \cdots \beta_n}^{\{k\}} \equiv \sum_{j = 1}^{n} Q_{\alpha_1 \cdots \alpha_m \beta_j \, , \, \beta_1 \cdots \beta_{j-1} \beta_{j+1} \cdots \beta_n}^{\{k-1\}} \nonumber\\[0.2cm]
= \sum_{j_1 = 1}^{n} \sum_{j_{2}=1 \atop j_{2} \neq j_{1}}^{n} \cdots \sum_{j_{k}=1 \atop j_{k} \neq j_{k-1}}^{n} Q_{\alpha_1 \cdots \alpha_m \beta_{j_1} \cdots \beta_{j_k} \, , \, (\boldsymbol{\cdot})} 
\label{Q_k_def}
\end{eqnarray}
for $k \geqslant 1$, with $Q^{(0)} \equiv Q$, and where we used the notation $(\boldsymbol{\cdot})$ to denote the set of all subscripts $\beta$ that are different of the subscripts $\beta$ present on the left of the comma, i.e. $\beta_{j_1} , \ldots , \beta_{j_k}$ here.   Therefore, substituting the result~\eref{A_partLbda_k_on_Q} into~\eref{action_T_on_Q_series_for_T}, exchanging the sums over $k$ and $n$ as
\begin{eqnarray}
\sum_{k = 0}^{\infty} \sum_{n=k}^{\infty} \, (\cdot) =\sum_{n = 0}^{\infty}  \sum_{k=0}^{\infty}  \, (\cdot) 
\end{eqnarray} 
and replacing $k$ by $j$, we find 
\begin{eqnarray}
\fl\hat T \left( \bm{A} \right) Q \left( \bm{\lambda} , \bm{A} \right) = \sum_{m = 0}^{\infty} \frac{1}{m!} \sum_{n = 0}^{\infty} \frac{1}{n!} \sum_{j=0}^{\infty} \frac{1}{j!} \, Q_{\alpha_1 \cdots \alpha_m \, , \, \beta_1 \cdots \beta_n}^{\{j\}} \lambda_{\alpha_1} \cdots \lambda_{\alpha_m} A_{\beta_1} \cdots A_{\beta_n} \, .
\label{T_on_Q_power_series_count_aff}
\end{eqnarray}

Now, we substitute the two power series~\eref{Q_minLbda_exp_count_par_and_aff} and~\eref{T_on_Q_power_series_count_aff} into the FR~\eref{fluct_rel_2} to get
\begin{eqnarray}
\fl \sum_{m = 0}^{\infty} \frac{(-1)^m}{m!} \sum_{n = 0}^{\infty} \frac{1}{n!} \, Q_{\alpha_1 \cdots \alpha_m \, , \, \beta_1 \cdots \beta_n} \lambda_{\alpha_1} \cdots \lambda_{\alpha_m} A_{\beta_1} \cdots A_{\beta_n} \nonumber\\[0.1cm]
= \sum_{m = 0}^{\infty} \frac{1}{m!} \sum_{n = 0}^{\infty} \frac{1}{n!} \sum_{j=0}^{\infty} \frac{1}{j!} \, Q_{\alpha_1 \cdots \alpha_m \, , \, \beta_1 \cdots \beta_n}^{\{j\}} \lambda_{\alpha_1} \cdots \lambda_{\alpha_m} A_{\beta_1} \cdots A_{\beta_n} \, ,
\label{fluct_rel_power_series}
\end{eqnarray}
which must thus be valid for any counting parameters $\bm{\lambda}$ and affinities $\bm{A}$. The sums over the index $m$ being exactly the same on both sides of~\eref{fluct_rel_power_series}, we can readily identify the coefficient of the $m^{\mathrm{th}}$ power of the counting parameters and those of the $n^{\mathrm{th}}$ power of the affinities. After multiplying both sides by $(-1)^m$, we get
\begin{eqnarray}
Q_{\alpha_1 \cdots \alpha_m \, , \, \beta_1 \cdots \beta_n} = (-1)^m \sum_{j=0}^{n} \frac{1}{j!} \, Q_{\alpha_1 \cdots \alpha_m \, , \, \beta_1 \cdots \beta_n}^{\{j\}} \, ,
\label{coef_m_power_count_n_power_aff}
\end{eqnarray}
 which is valid for any $m,n \geqslant 0$.

Now, remember that the quantity $Q^{\{j\}}$ is defined by~\eref{Q_k_def}, i.e. the sum
\begin{eqnarray}
Q_{\alpha_1 \cdots \alpha_m \, , \, \beta_1 \cdots \beta_n}^{\{j\}} = \sum_{k_1 = 1}^{n} \sum_{k_{2}=1 \atop k_{2} \neq k_{1}}^{n} \cdots \sum_{k_{j}=1 \atop k_{j} \neq k_{j-1}}^{n} Q_{\alpha_1 \cdots \alpha_m \beta_{k_1} \cdots \beta_{k_j} \, , \, (\boldsymbol{\cdot})}
\label{Q_j_def}
\end{eqnarray}
for any $1 \leqslant j \leqslant n$, and $Q^{\{0\}}$ being merely $Q$ itself. Remembering the invariance~\eref{Q_inv_perm_indices}, $Q^{\{j\}}$ can then be rewritten in a more convenient form by identifying the total number of terms $Q_{\alpha_1 \cdots \alpha_m \beta_{l_1} \cdots \beta_{l_j} \, , \, (\boldsymbol{\cdot})}$, with $l_1 < \cdots < l_j$, that occur in the right-hand side of~\eref{Q_j_def}. Note that the latter involves sums over indices $k_1, \ldots , k_j$ that must all take different values. This ensures that counting all the $Q_{\alpha_1 \cdots \alpha_m \beta_{l_1} \cdots \beta_{l_j} \, , \, (\boldsymbol{\cdot})}$ such that $l_1 < \cdots < l_j$ indeed includes all the terms present in the right-hand side of~\eref{Q_j_def}. Furthermore, the total number of such terms $Q_{\alpha_1 \cdots \alpha_m \beta_{l_1} \cdots \beta_{l_j} \, , \, (\boldsymbol{\cdot})}$ can be readily seen to be nothing but the total number of permutations of the $j$ distinct values $l_1 , \ldots , l_j$, that is merely $j!$. Therefore, the quantity $Q^{\{j\}}$ can be alternatively written as
\begin{eqnarray}
Q_{\alpha_1 \cdots \alpha_m \, , \, \beta_1 \cdots \beta_n}^{\{j\}} = j! \, Q_{\alpha_1 \cdots \alpha_m \, , \, \beta_1 \cdots \beta_n}^{(j)}
\label{Q_j_alt_expr}
\end{eqnarray}
in terms of~\eref{Q_j_expr}. Substituting this expression into~\eref{coef_m_power_count_n_power_aff} then yields~\eref{gen_rel_response_coef}. Q.E.D.

\section*{References}

\bibliographystyle{unsrt}
\bibliography{DATABASE_Nonlinear_response}

\begin{thebibliography}{10}

\bibitem{deGr}
S.~R. de~Groot and P.~Mazur.
\newblock {\em Non-Equilibrium Thermodynamics}.
\newblock Dover, New York, 1984.

\bibitem{Diu}
B.~Diu, C.~Guthmann, D.~Lederer, and B.~Roulet.
\newblock {\em Physique Statistique}.
\newblock Hermann, Paris, 1997.

\bibitem{Kle55}
M.~J. Klein.
\newblock Principle of detailed balance.
\newblock {\em Phys. Rev.}, 97:1446, 1955.

\bibitem{DeDon}
T.~de~Donder and P.~Van Rysselberghe.
\newblock {\em Affinity}.
\newblock Stanford University Press, Menlo Park CA, 1936.

\bibitem{Prig}
I.~Prigogine.
\newblock {\em Introduction to Thermodynamics of Irreversible Processes}.
\newblock Wiley, New York, 1967.

\bibitem{Cal}
H.~B. Callen.
\newblock {\em Thermodynamics and An Introduction to Thermostatistics}.
\newblock Wiley, New York, 1985.

\bibitem{Gre52}
M.~S. Green.
\newblock Markoff random processes and the statistical mechanics of
  time-dependent phenomena.
\newblock {\em J. Chem. Phys.}, 20:1281, 1952.

\bibitem{Gre54}
M.~S. Green.
\newblock Markoff random processes and the statistical mechanics of
  time-dependent phenomena. {II}. {I}rreversible processes in fluids.
\newblock {\em J. Chem. Phys.}, 22:398, 1954.

\bibitem{Kub57}
R.~Kubo.
\newblock Statistical-mechanical theory of irreversible processes. {I}.
  {G}eneral theory and simple applications to magnetic and conduction problems.
\newblock {\em J. Phys. Soc. Jpn.}, 12:570, 1957.

\bibitem{CW51}
H.~B. Callen and T.~A. Welton.
\newblock Irreversibility and generalized noise.
\newblock {\em Phys. Rev.}, 83:34, 1951.

\bibitem{Kub66}
R.~Kubo.
\newblock The fluctuation-dissipation theorem.
\newblock {\em Rep. Prog. Phys.}, 29:255, 1966.

\bibitem{Ons31a}
L.~Onsager.
\newblock Reciprocal relations in irreversible processes. {I}.
\newblock {\em Phys. Rev.}, 37:405, 1931.

\bibitem{Ons31b}
L.~Onsager.
\newblock Reciprocal relations in irreversible processes. {II}.
\newblock {\em Phys. Rev.}, 38:2265, 1931.

\bibitem{Cas45}
H.~B.~G. Casimir.
\newblock On {O}nsager's principle of microscopic reversibility.
\newblock {\em Rev. Mod. Phys.}, 17:343, 1945.

\bibitem{ECM93}
D.~J. Evans, E.~G.~D. Cohen, and G.~P. Morriss.
\newblock Probability of second law violations in shearing steady states.
\newblock {\em Phys. Rev. Lett.}, 71:2401, 1993.

\bibitem{ES94}
D.~J. Evans and D.~J. Searles.
\newblock Equilibrium microstates which generate second law violating steady
  states.
\newblock {\em Phys. Rev. E}, 50:1645, 1994.

\bibitem{GC95}
G.~Gallavotti and E.~G.~D. Cohen.
\newblock Dynamical ensembles in nonequilibrium statistical mechanics.
\newblock {\em Phys. Rev. Lett.}, 74:2694, 1995.

\bibitem{Jar97}
C.~Jarzynski.
\newblock Nonequilibrium equality for free energy differences.
\newblock {\em Phys. Rev. Lett.}, 78:2690, 1997.

\bibitem{Kur98}
J.~Kurchan.
\newblock Fluctuation theorem for stochastic dynamics.
\newblock {\em J. Phys. A: Math. Gen.}, 31:3719, 1998.

\bibitem{LS99}
J.~L. Lebowitz and H.~Spohn.
\newblock A {G}allavotti-–{C}ohen-type symmetry in the large deviation
  functional for stochastic dynamics.
\newblock {\em J. Stat. Phys.}, 95:333, 1999.

\bibitem{Cro99}
G.~E. Crooks.
\newblock Entropy production fluctuation theorem and the nonequilibrium work
  relation for free energy differences.
\newblock {\em Phys. Rev. E}, 60:2721, 1999.

\bibitem{ZC03}
R.~van Zon and E.~G.~D. Cohen.
\newblock Extension of the fluctuation theorem.
\newblock {\em Phys. Rev. Lett.}, 91:110601, 2003.

\bibitem{AGM09}
D.~Andrieux, P.~Gaspard, T.~Monnai, and S.~Tasaki.
\newblock The fluctuation theorem for currents in open quantum systems.
\newblock {\em New J. Phys.}, 11:043014, 2009.

\bibitem{EHM09}
M.~Esposito, U.~Harbola, and S.~Mukamel.
\newblock Nonequilibrium fluctuation, fluctuation theorems and counting
  statistics in quantum systems.
\newblock {\em Rev. Mod. Phys.}, 81:1665, 2009.

\bibitem{CHT11}
M.~Campisi, P.~H\"anggi, and P.~Talkner.
\newblock Colloquium: {Q}uantum fluctuation relations: {F}oundations and
  applications.
\newblock {\em Rev. Mod. Phys.}, 83:771, 2011.

\bibitem{HPPG11}
P.~I. Hurtado, C.~P\'erez-Espigares, J.~J. del Pozo, and P.~L. Garrido.
\newblock Symmetries in fluctuations far from equilibrium.
\newblock {\em Proc. Natl. Acad. Sci. U.S.A.}, 108:7704, 2011.

\bibitem{Gas13_1}
P.~Gaspard.
\newblock Multivariate fluctuation relations for currents.
\newblock {\em New J. Phys.}, 15:115014, 2013.

\bibitem{Gas13_2}
P.~Gaspard.
\newblock Time-reversal symmetry relations for currents in quantum and
  stochastic nonequilibrium systems.
\newblock In R.~Klages, W.~Just, and C.~Jarzynski, editors, {\em Nonequilibrium
  Statistical Physics of Small Systems: Fluctuation Relations and Beyond},
  pages 213--257. Wiley-VCH Verlag, Weinheim, 2013.

\bibitem{Gal96}
G.~Gallavotti.
\newblock Extension of {O}nsager's reciprocity to large fields and the chaotic
  hypothesis.
\newblock {\em Phys. Rev. Lett.}, 77:4334, 1996.

\bibitem{AG04}
D.~Andrieux and P.~Gaspard.
\newblock Fluctuation theorem and {O}nsager reciprocity relations.
\newblock {\em J. Chem. Phys.}, 121:6167, 2004.

\bibitem{AG07}
D.~Andrieux and P.~Gaspard.
\newblock A fluctuation theorem for currents and non-linear response
  coefficients.
\newblock {\em J. Stat. Mech.}, 2007:02006, 2007.

\bibitem{SU08}
K.~Saito and Y.~Utsumi.
\newblock Symmetry in full counting statistics, fluctuation theorem, and
  relations among nonlinear transport coefficients in the presence of a
  magnetic field.
\newblock {\em Phys. Rev. B}, 78:115429, 2008.

\bibitem{WF15}
C.~Wang and D.~E. Feldman.
\newblock Fluctuation relations for spin currents.
\newblock {\em Phys. Rev. B}, 92:064406, 2015.

\bibitem{JW04}
C.~Jarzynski and D.~K. W\'ojcik.
\newblock Classical and quantum fluctuation theorems for heat exchange.
\newblock {\em Phys. Rev. Lett.}, 92:230602, 2004.

\bibitem{AG06}
D.~Andrieux and P.~Gaspard.
\newblock Fluctuation theorem for transport in mesoscopic systems.
\newblock {\em J. Stat. Mech.}, 2006:P01011, 2006.

\bibitem{Gas04}
P.~Gaspard.
\newblock Fluctuation theorem for nonequilibrium reactions.
\newblock {\em J. Chem. Phys.}, 120:8898, 2004.

\bibitem{Naz}
Y.~V. Nazarov and Y.~M. Blanter.
\newblock {\em Quantum Transport: Introduction to Nanoscience}.
\newblock Cambridge University Press, Cambridge UK, 2009.

\bibitem{vanK}
N.~G. van Kampen.
\newblock {\em Stochastic Processes in Physics and Chemistry}.
\newblock North Holland, Amsterdam, 3rd edition, 2007.

\bibitem{Kur00}
J.~Kurchan.
\newblock A quantum fluctuation theorem.
\newblock {\em ar{X}iv:cond-mat/0007360}, 2000.

\bibitem{L57}
R.~Landauer.
\newblock Spatial variation of currents and fields due to localized scatterers
  in metallic conduction.
\newblock {\em IBM J. Res. Dev.}, 1:223, 1957.

\bibitem{L70}
R.~Landauer.
\newblock Electrical resistance of disordered one-dimensional lattices.
\newblock {\em Phil. Mag.}, 21:863, 1970.

\bibitem{B86a}
M.~B\"uttiker.
\newblock Role of quantum coherence in series resistors.
\newblock {\em Phys. Rev. B}, 33:3020, 1986.

\bibitem{B86b}
M.~B\"uttiker.
\newblock Four-terminal phase-coherent conductance.
\newblock {\em Phys. Rev. Lett.}, 57:1761, 1986.

\bibitem{LL93}
L.~S. Levitov and G.~B. Lesovik.
\newblock Charge distribution in quantum shot noise.
\newblock {\em JETP Lett.}, 58:230, 1993.

\bibitem{LS11}
G.~B. Lesovik and I.~A. Sadovskyy.
\newblock Scattering matrix approach to the description of quantum electron
  transport.
\newblock {\em Phys.-Usp.}, 54:1007, 2011.

\bibitem{AndPhD}
D.~Andrieux.
\newblock {\em Nonequilibrium Statistical Thermodynamics at the Nanoscale: From
  {M}axwell demon to biological information processing ({P}h{D} thesis,
  Universit\'e libre de Bruxelles, 2008)}.
\newblock VDM Verlag, ISBN 978-3639139334, 2009.

\bibitem{Rys}
H.~J. Ryser.
\newblock {\em Combinatorial Mathematics}.
\newblock The Mathematical Association of America, Buffalo NY, 1963.

\bibitem{GradRyz}
I.~S. Gradshteyn and I.~M. Ryzhik.
\newblock {\em Table of Integrals, Series, and Products}.
\newblock Elsevier/Academic, Amsterdam, 7th edition, 2007.

\bibitem{AbrSteg}
M.~Abramowitz and I.~A. Stegun.
\newblock {\em Handbook of Mathematical Functions with Formulas, Graphs, and
  Mathematical Tables}.
\newblock National Bureau of Standards, Applied Mathematics Series, 55, 1964.

\bibitem{Apo99}
T.~M. Apostol.
\newblock An elementary view of {E}uler's summation formula.
\newblock {\em Amer. Math. Monthly}, 106:409, 1999.

\end{thebibliography}


\end{document}